\documentclass[11pt]{article}
\usepackage[letterpaper, margin=1in]{geometry}

\usepackage{lmodern}
\usepackage{microtype}

\usepackage{booktabs}

\usepackage{amsmath, amssymb, amsthm}

\theoremstyle{plain}
\newtheorem{theorem}{Theorem}[section]
\newtheorem{lemma}{Lemma}[section]
\newtheorem{proposition}{Proposition}[section]
\newtheorem{corollary}{Corollary}[section]
\newtheorem{fact}{Fact}[section]
\newtheorem{remark}{Remark}[section]
\theoremstyle{definition}
\newtheorem{definition}{Definition}[section]

\usepackage{algorithm}
\usepackage{algpseudocode}
\algrenewcommand\algorithmicrequire{\textbf{Input:}}
\algrenewcommand\algorithmicensure{\textbf{Output:}}

\usepackage{hyperref}
\hypersetup{
  colorlinks = true, 
  urlcolor = blue, 
  linkcolor = magenta, 
  citecolor = blue 
}

\makeatletter
\AddToHook{env/algorithmic/before}{\def\@currentcounter{ALG@line}}

\renewcommand{\theHALG@line}{\thealgorithm.\arabic{ALG@line}}
\makeatother

\usepackage[nameinlink]{cleveref}

\crefalias{ALG@line}{line}

\crefname{fact}{fact}{facts}
\Crefname{fact}{Fact}{Facts}

\usepackage[svgnames]{xcolor}

\usepackage{natbib}

\renewcommand{\natural}{\mathbb N}
\newcommand{\integer}{\mathbb Z}
\newcommand{\real}{\mathbb R}

\newcommand{\calA}{\mathcal{A}}

\newcommand{\calC}{\mathcal{C}}
\newcommand{\calD}{\mathcal{D}}

\newcommand{\calF}{\mathcal{F}}
\newcommand{\calG}{\mathcal{G}}
\newcommand{\calH}{\mathcal{H}}

\newcommand{\calL}{\mathcal{L}}

\newcommand{\calN}{\mathcal{N}}

\newcommand{\calP}{\mathcal{P}}

\newcommand{\calU}{\mathcal{U}}

\newcommand{\calX}{\mathcal{X}}
\newcommand{\calY}{\mathcal{Y}}

\newcommand{\indicator}[1]{\mathbf{1}\{#1\}}
\newcommand{\INDICATOR}[1]{\mathbf{1}\left\{#1\right\}}

\newcommand{\abs}[1]{\lvert #1 \rvert}
\newcommand{\ABS}[1]{\left\lvert #1 \right\rvert}

\DeclareMathOperator*{\probability}{\mathbf{Pr}}
\DeclareMathOperator*{\expectation}{\mathbf{E}}

\newcommand{\pr}[2]{\probability_{#1}[#2]}
\newcommand{\PR}[2]{\probability_{#1}\left[#2\right]}
\newcommand{\ev}[2]{\expectation_{#1}[#2]}
\newcommand{\EV}[2]{\expectation_{#1}\left[#2\right]}

\newcommand{\pup}{p_\mathrm{up}}
\newcommand{\pdown}{p_\mathrm{down}}

\DeclareMathOperator*{\E}{\mathbf{E}}

\newcommand{\Dtrain}{\calD^{\mathrm{train}}}
\newcommand{\Dtest}{\calD^{\mathrm{test}}}
\newcommand{\errtrain}{\mathrm{err}_{\mathrm{train}}}
\newcommand{\errtest}{\mathrm{err}_{\mathrm{test}}}

\title{Iterative Chow Filtering for Learning with Distribution Shift}
\author{
    Gautam Chandrasekaran\thanks{All authors were supported in part by the NSF AI Institute for Foundations of Machine Learning (IFML).} \\ \texttt{gautamc@cs.utexas.edu} \\ UT Austin
    \and
    Georgios Gkrinias\footnotemark[1] \\ \texttt{gkrinias@cs.utexas.edu} \\ UT Austin
    \and 
    Adam R. Klivans\footnotemark[1] \\ \texttt{klivans@cs.utexas.edu} \\ UT Austin 
    \and
    Konstantinos Stavropoulos\footnotemark[1] \thanks{Also supported by the 2025 Apple Scholars in AI/ML PhD fellowship.} \\ \texttt{kstavrop@utexas.edu} \\ UT Austin
    \and
    Arsen Vasilyan\footnotemark[1] \\ \texttt{arsenvasilyan@gmail.com} \\ UT Austin
}
\date{}

\begin{document}

\maketitle

\begin{abstract}
Recent work due to Goel et al.~gave the first efficient algorithms for learning with distribution shift in the challenging PQ framework.  In this setting, a learner receives labeled training examples, unlabeled test examples, and must make correct predictions on the test set but is allowed to abstain from predicting on out-of-distribution points.  Their results rely on $\calL_2$ sandwiching  approximations, a strong requirement that leads to poor bounds for several basic function classes such as DNF formulas.

Here, we show that the weaker notion of $\calL_1$ sandwiching suffices for efficient PQ learning. As a consequence, we obtain the first quasipolynomial-time PQ learning algorithm for DNFs under the uniform distribution and essentially match the guarantees known for ordinary PAC learning. More broadly, our bounds provide exponential improvements for several classes including constant depth circuits and constant degree polynomial threshold functions.

Our main technical ingredient is \emph{Iterative Chow Filtering}, a new procedure that uses low-degree Chow parameters to identify and remove test points incompatible with the training distribution.
\end{abstract}

\newpage

\section{Introduction}

A central challenge in modern machine learning is to obtain reliability in the presence of distribution shift. Arguably the most well-studied formalization of the problem is domain adaptation~\citep{ben2006analysis,blitzer2007learning,mansour2009domadapt,ben2010theory,david2010impossibility}, where a learner observes labeled examples from a training distribution, unlabeled examples from a test distribution, and must output a hypothesis with low test error (despite never observing test labels). In the simplest version (often referred to as  covariate shift), the training and test labels are generated by the same unknown ground truth function $f^*$.  Despite more than two decades of work, most known guarantees in this literature are information-theoretic and are not accompanied by efficient algorithms.

This gap has motivated a recent line of work on establishing efficient algorithms for learning with distribution shift in the TDS model due to Klivans et al. \cite{klivans2024testablelearningdistributionshift} and the more fine-grained PQ framework of Goldwasser et al.~\cite{goldwasser2020beyond}. Most relevant to our work here is PQ learning, in which a learner is allowed to abstain on test points that appear out of distribution but must be accurate on the points on which it does make predictions. Formally, a PQ learner outputs a hypothesis $h$ and a selector $s$ such that, under the test distribution, the probability that $s(x)=1$ and $h(x)\neq f^*(x)$ is small, while under the training distribution, the probability that $s(x)=0$ is small. Thus, a PQ learner must provide accurate predictions on the portion of the test distribution that remains sufficiently similar to the training distribution.

Goel et al.~\cite{goel2024tolerantalgorithmslearningarbitrary} obtained the first efficient PQ learning algorithms for a variety of function classes, but several basic cases remained open.  For example, no subexponential-time PQ algorithm was known for learning DNF formulas with respect to the uniform distribution on the hypercube, one of the most heavily studied problems in computational learning theory. This is in sharp contrast with the standard realizable PAC setting where  the seminal result of Linial, Mansour, and Nisan~\cite{lmnac0} gives a quasipolynomial-time algorithm for learning constant-depth circuits under the uniform distribution via low-degree polynomial approximations.

The obstacle in the PQ setting is that known algorithms require a stronger form of polynomial approximation. The result of~\cite{goel2024tolerantalgorithmslearningarbitrary}, while also based on polynomial approximation, relies on $\calL_2$ sandwiching: for the target function $f^*$, one must construct polynomials $\pdown,\pup$ satisfying $\pdown(x)\le f^*(x)\le \pup(x)$ pointwise, and such that the expected squared gap $(\pup(x)-\pdown(x))^2$ under the training distribution is small. For DNFs, however, the best known $\calL_2$ sandwiching bounds are exponentially worse than the corresponding bounds for ordinary polynomial approximation, with degree bounds depending polynomially on the size.

In this work, we show that $\calL_1$ sandwiching suffices for efficient PQ learning. Here, the pointwise sandwiching requirement is the same, but the approximation guarantee only asks that the expected gap $\pup(x)-\pdown(x)$ be small, rather than its square. This distinction is crucial: by the celebrated result of Braverman~\cite{10.1145/1754399.1754401}, size-$s$ constant-depth circuits have $\calL_1$ sandwiching degree $\mathrm{poly}(\log s)$, matching the qualitative behavior of standard polynomial approximation. This essentially closes the gap between the known computational complexity of PAC learning and PQ learning for DNFs under the uniform distribution, yielding the first quasipolynomial-time PQ learning algorithm in this setting.
More broadly, our result implies exponential runtime improvements for several well-studied function classes (see \Cref{table:results-l1-sandwiching}). 

Our main technical ingredient is a new outlier-filtering procedure, which we call \emph{Iterative Chow Filtering} (ICF). The procedure combines ideas from filtering methods in robust learning and learning with distribution shift with the classical notion of Chow parameters, allowing us to convert $\calL_1$ sandwiching into the selective, pointwise coverage guarantees required by PQ learning.

\subsection{Our Results}

Our main result is a new computationally efficient algorithm for PQ learning that applies to any class admitting low-degree $\calL_1$ sandwiching polynomials. In \Cref{table:results-l1-sandwiching}, we present a number of concrete applications of our main result by combining known sandwiching-degree bounds for the corresponding classes.

\begin{table*}[!htbp]\begin{center}
\begin{tabular}{c c c c c} 
 \toprule
 & \textbf{Concept class} $\calC$ &  \textbf{Training Marginal} & \textbf{Our Work} & \textbf{Prior Work} \\ \midrule
1 & Degree-2 PTFs & 
\begin{tabular}{c}
    Uniform on $\{\pm 1\}^d$
\end{tabular}  & $\mathrm{poly}(d)$ & $2^{O(d)}$ \\ \midrule
2 & Constant-Degree PTFs & Standard Gaussian & $\mathrm{poly}(d)$ & $2^{\mathrm{poly}(d)}$ \\ \midrule
3 & \begin{tabular}{c} Functions of a Constant \\ Number of Halfspaces\end{tabular} & Isotropic Log-Concave & $\mathrm{poly}(d)$ & None \\ \midrule
4 & DNFs of Size $s$
 & Uniform on $\{\pm 1\}^d$   & $d^{ O(\log^2({s}))}$ & \begin{tabular}{c}$d^{\widetilde{O}(\sqrt{s})}$ \end{tabular} \\ \midrule
5 & Circuits 
of Size $s$, Depth $t$ 
 & Uniform on $\{\pm 1\}^d$   & $d^{ (\log({s}))^{O(t)}}$ & \begin{tabular}{c}$d^{\sqrt{s}\cdot (\log({s}))^{O(t)}}$ \\ {\footnotesize{only for formulas}} \end{tabular} \\
 \bottomrule
\end{tabular}
\end{center}
\caption{Runtime comparison between our PQ learning algorithm and the best previous results for PQ learning various classes up to constant error and rejection rate, even when both the training and the test labels are given by the same (unknown) function $f^*\in \calC$, i.e., when there is only covariate shift.}
\label{table:results-l1-sandwiching}
\end{table*}

More formally, we consider a feature space $\calX\subseteq \real^d$ and binary labels $\calY=\{0,1\}$. The learner has access to labeled samples from a training distribution $\calD^{\mathrm{train}}$ over $\calX\times\calY$, as well as unlabeled samples from the marginal $\calD^{\mathrm{test}}_\calX$ of a test distribution $\calD^{\mathrm{test}}$. In the general case captured by our results, neither distribution need be realizable. Since the learner has no access to test labels, the standard benchmark from domain adaptation \cite{ben2006analysis,blitzer2007learning} is the joint optimal error
\begin{equation}
    \lambda
    =
    \lambda(\calD^{\mathrm{train}},\calD^{\mathrm{test}};\calC)
    =
    \min_{f\in \calC}
    \Bigl(
        \errtrain(f)+\errtest(f)
    \Bigr),
    \label{equation:definition-lambda}
\end{equation}
where $\errtrain(f)=\pr{(x,y)\sim\Dtrain}{f(x)\neq y}$ and $\errtest(f)$ is defined analogously.

It will be useful to separate this benchmark into its training and test components. Let
\[
    \calC_{\lambda}
    =
    \{f\in\calC:\errtrain(f)+\errtest(f)=\lambda\}.
\]
We define
\[
    \lambda_{\mathrm{train}}
    =
    \min_{f\in\calC_\lambda}\errtrain(f)
    \qquad\text{and}\qquad
    \lambda_{\mathrm{test}}
    =
    \lambda-\lambda_{\mathrm{train}}.
\]
Thus, $\lambda_{\mathrm{train}}$ and $\lambda_{\mathrm{test}}$ correspond to the training and test errors of a joint optimizer, with ties broken in favor of smaller training error.

The formal definition of PQ learning of a concept class $\calC\subseteq\{\calX\to \calY\}$ with respect to a distribution $\calD$ over $\calX$ is as follows.
\begin{definition}
[PQ Learning, modification of \cite{goldwasser2020beyond}]
\label{def:pq}
Let $\epsilon,\delta,\eta\in(0,1)$. An algorithm $\calA$ is an $(\epsilon,\delta,\eta)$-PQ learner for $\calC$ with respect to $\calD$ if the following holds whenever $\Dtrain_\calX=\calD$. Given sufficiently many labeled samples from $\Dtrain$ and sufficiently many unlabeled samples from $\Dtest_\calX$, the algorithm outputs a classifier $h \colon\calX\to\{0,1\}$ and a selector $s \colon\calX\to\{0,1\}$ such that, with probability at least $1-\delta$:
\begin{enumerate}
    \item \emph{Selective accuracy.}
    The error on accepted test points satisfies
    \[
    \pr{(x,y)\sim\Dtest}{h(x)\neq y \land s(x)=1}
    \leq
    \lambda_{\mathrm{test}}
    +
    \frac{1}{\eta}
    \Bigl(
        \lambda_{\mathrm{train}}
        +
        \min_{f\in\calC}\errtrain(f)
    \Bigr)
    +
    \epsilon.
    \]
    \item \emph{Rejection rate.}
    The rejection probability under the training marginal satisfies
    \[
    \pr{x\sim\Dtrain_\calX}{s(x)=0}\leq \eta.
    \]
\end{enumerate}
\end{definition}

\begin{remark}
    \cite{goldwasser2020beyond} showed that there are simple examples where any PQ learner must either have accuracy $\Omega(\sqrt{\lambda_{\mathrm{train}}})$ or rejection rate $\Omega(\sqrt{\lambda_{\mathrm{train}}})$.
\end{remark}

We now provide a definition for the sandwiching degree of the concept class $\calC$ with respect to $\calD$.

\begin{definition}[\(\calL_1\) Sandwiching Degree] \label{def:l1_sandwiching} 
    We let $\calL_1(\calC,\calD,\epsilon)$ denote the minimum natural number $\ell$ such that for any $f\in\calC$, there are polynomials \(p_\mathrm{up}, p_\mathrm{down}\) of degree at most $\ell$ such that:
    \[ p_\mathrm{down}(x) \leq f(x) \leq p_\mathrm{up}(x) \text{ for all }x \in \calX \quad\text{ and }\quad
    \ev{x \sim \calD}{p_\mathrm{up}(x) - p_\mathrm{down}(x)} \leq \epsilon \, .\]
\end{definition}

Our main result is the following theorem, whose formal version can be found in \Cref{thm:pq}.

\begin{theorem}[PQ Learning via \(\calL_1\) Sandwiching] \label{thm:pq-intro} Let any \(\epsilon, \eta, \delta \in (0, 1)\) and suppose \(\calD\) is hypercontractive. Let $\ell = \calL_1(\calC,\calD,\epsilon\eta/c)$ for some appropriately large universal constant $c\ge 1$. There is an $(\epsilon,\delta,\eta)$-PQ learner for \(\calC\) with respect to \(\calD\) with sample complexity and runtime \(\mathrm{poly} (d^\ell, 1/\epsilon, 1/\eta, \log(1/\delta)^\ell )\).
\end{theorem}

\paragraph{Relaxed Requirements.}
As mentioned earlier, our main result yields significant runtime improvements for several basic classes, including those in \Cref{table:results-l1-sandwiching}. This is because $\calL_1$ sandwiching is a relaxed requirement compared to $\calL_2$ sandwiching, but also because our result imposes no coefficient bounds on the sandwiching polynomials $\pup,\pdown$. The latter improvement is important for the class of functions of a constant number of halfspaces under any fixed isotropic log-concave training distribution, where the available sandwiching approximators \cite{KKM13} are not known to have bounded coefficients. In contrast, \cite{goel2024tolerantalgorithmslearningarbitrary} require both an $\calL_2$ sandwiching guarantee and explicit bounds on the coefficients of the sandwiching polynomials.\footnote{Note that both \cite{goel2024tolerantalgorithmslearningarbitrary} and our main result require the additional assumption that the training marginal is hypercontractive. This condition is important for the concentration of polynomials and is satisfied by a wide range of distributions, including any log-concave distribution, as well as any product distribution over the Boolean hypercube.}

\paragraph{Improved Error Guarantees.}
Our algorithm also achieves sharper error guarantees on the accepted portion of the test distribution, expressed in terms of the label-drift parameters $\lambda_{\mathrm{train}}$ and $\lambda_{\mathrm{test}}$, together with the optimal training error. Compared with the previous best result of \cite{goel2024tolerantalgorithmslearningarbitrary}, which gives an $O(\lambda/\eta)$-type guarantee with a constant-factor loss in the leading term, our analysis obtains the sharpest known dependence on these parameters (as stated in \Cref{def:pq}). In particular, this improves the best known error guarantees even for the fundamental case of PQ learning linear separators under Gaussian training marginals.

\paragraph{Tolerant TDS Learning.}
A related framework is TDS learning, where the learner may abstain on the entire test distribution, provided it detects the presence of distribution shift (it is easy to see that PQ learning implies TDS learning). Chandrasekaran et al.~\cite{chandrasekaran2024efficient} showed that $\calL_1$ sandwiching suffices for efficient TDS learning. Their result, however, does not extend to tolerant TDS learning, where the learner is required to accept whenever the training and test marginals are close in total variation distance. Our PQ learning theorem implies that $\calL_1$ sandwiching suffices for tolerant TDS learning, but we can obtain even sharper error bounds by giving a direct white-box analysis of our algorithm for the tolerant setting (see \Cref{appendix:tolerant_tds}).

\subsection{Our Techniques}

We prove our main PQ-learning result in two steps. First, using the training data, we learn a hypothesis $\hat{f} \colon\calX\to\{0,1\}$ with small training error via $\calL_1$ polynomial regression. This is possible because the class $\calC$ admits low-degree $\calL_1$ sandwiching polynomials. In fact, the classical work of \cite{doi:10.1137/060649057} shows that even the weaker existence of approximating polynomials suffices for efficient learning with near-optimal training error via $\calL_1$ regression.

The main challenge is to identify the part of the test distribution on which $\hat{f}$ can be trusted. We construct a selector $s \colon\calX\to\{0,1\}$ that accepts test points on which $\hat{f}$ has small error, while rejecting only a small fraction of the training distribution. The key ingredient is an error-transfer lemma: we identify a condition under which the training error of $\hat{f}$ transfers to the selected test distribution, and then design an efficient algorithm that enforces this condition.

Existing transfer lemmas for learning under distribution shift are insufficient for our purposes. They either require stronger approximation guarantees, such as $\calL_2$ sandwiching rather than $\calL_1$ sandwiching, or they do not yield an efficient method for constructing the selector.

The closest prior result is the following lemma of \cite{chandrasekaran2024efficient}, proved in the context of TDS learning.

\begin{lemma}[Informal, from \cite{chandrasekaran2024efficient}]
    Suppose that $\calC$ has $\calL_1$ sandwiching polynomials of degree $\ell$ with respect to the training distribution, and let $\hat{f} \colon\calX\to\{0,1\}$. Suppose also that for every monomial $\chi$ of degree at most $\ell$,
    \[
        \E_{\calD^{\mathrm{train}}}[\hat{f}(x)\chi(x)]
        \approx
        \E_{\calD^{\mathrm{test}}}[\hat{f}(x)\chi(x)]
        \quad\text{and}\quad
        \E_{\calD^{\mathrm{train}}}[\chi(x)]
        \approx
        \E_{\calD^{\mathrm{test}}}[\chi(x)].
    \]
    Then, the test error of $\hat{f}$ is approximately upper bounded by
    $\lambda(\calC)+\mathrm{err}_{\mathrm{train}}(\hat{f})$.
\end{lemma}

The quantities $\E[\hat{f}(x)\chi(x)]$ are the so-called Chow parameters of $\hat{f}$ \cite{chow-original,chow-ltf,chow-ptf}. Thus, this lemma shows that, under $\calL_1$ sandwiching, matching the low-degree Chow parameters of $\hat{f}$ and of the constant function suffices for error transfer in the (weaker) model of TDS learning, where a learner can abstain on the entire test set.

However, requiring additive Chow matching between a {\em selection} of the test distribution and the training distribution is too much to hope for. Consider $\calX=\mathbb{R}$, $\hat{f}(x)=\indicator{x\ge 0}$, $\calD^{\mathrm{train}}=\calU[-1,1]$, and $\calD^{\mathrm{test}}=\calU[0,1]$. Any selector $s$ that rejects only a small fraction of the training distribution must accept almost all of $[0,1]$. Hence,
\[
    \E_{\calD^{\mathrm{test}}}[\hat{f}(x)s(x)] \approx 1,
    \qquad\text{whereas}\qquad
    \E_{\calD^{\mathrm{train}}}[\hat{f}(x)] = \frac12.
\]
Thus, no valid selector can generally make the selected test distribution Chow-match the training distribution.

We overcome this obstruction with a multiplicative Chow-transfer lemma.

\begin{lemma}[Informal, see \Cref{lemma:chow_transfer}]
    Suppose that $\calC$ has $\calL_1$ sandwiching polynomials of degree $\ell$ with respect to the training distribution, and let $\hat{f} \colon\calX\to\{0,1\}$. Also suppose that, for some $R > 1$, every nonnegative polynomial $p$ of degree at most $\ell$ satisfies
    \[
        \E_{\calD^{\mathrm{test}}}[\hat{f}(x)p(x)]
        \lesssim
        R\cdot \E_{\calD^{\mathrm{train}}}[\hat{f}(x)p(x)] \quad\text{ and }\quad
        \E_{\calD^{\mathrm{test}}}[(1-\hat{f}(x))p(x)]
        \lesssim
        R\cdot \E_{\calD^{\mathrm{train}}}[(1-\hat{f}(x))p(x)].
    \]
    Then, the test error of $\hat{f}$ is approximately upper bounded by
    $R\cdot\left(\lambda(\calC)+\mathrm{err}_{\mathrm{train}}(\hat{f})\right)$.
\end{lemma}
Instead of requiring additive matching of Chow parameters, this lemma only requires multiplicative domination of the Chow-type expectations of $\hat{f}$ and $1-\hat{f}$ against nonnegative low-degree polynomial tests. We enforce this condition on the selected test distribution using an efficient algorithm, which we call Iterative Chow Filtering (ICF).

ICF proceeds by repeatedly searching for a violated Chow-domination inequality. Suppose, for example, that there is a polynomial $p$ such that the expectation of $\hat{f}(x)p(x)$ under the current test distribution is much larger than the expectation of $\hat{f}(x)|p(x)|$ under the training distribution. Then, for an appropriate threshold $\tau$, the event
\[
    \hat{f}(x)|p(x)| > \tau
\]
has significantly larger mass under the current test distribution than under the training distribution. We reject all points in this event and recurse, applying the same argument symmetrically with $1-\hat{f}$ in place of $\hat{f}$. Each step removes disproportionately more test mass than training mass, so the total rejected training mass remains small.

Iterative filtering is a standard paradigm in robust learning \cite{klivans2009learning,diakonikolas2018learning,diakonikolas2019robust,diakonikolas2019sever,klivans2024learningac0,klivans2026poweriterativefilteringsupervised} and has also appeared in prior work on PQ learning \cite{goel2024tolerantalgorithmslearningarbitrary}. The key difference is that all previous filtering procedures are based on discrepancies in polynomial expectations of only the marginal distributions. Our procedure here instead filters using discrepancies in Chow expectations, namely polynomial expectations weighted by the learned Boolean hypothesis $\hat{f}$ or by $1-\hat{f}$. This Chow-based filtering step is what allows us to obtain PQ learning from $\calL_1$ sandwiching alone.

\subsection{Related Work}

\paragraph{PQ Learning.}
The framework of PQ learning was introduced by \cite{goldwasser2020beyond}. In its original formulation, they showed that, without assumptions on the training distribution, an agnostic empirical risk minimization oracle yields PQ learning via an oracle-efficient reduction. This was later strengthened by \cite{kalai21acovariateshift}, who showed that distribution-free PQ learning is equivalent to the reliable agnostic learning framework of \cite{KALAI2012reliable}. Reliable agnostic learning lies between realizable PAC learning and agnostic learning: it is easier than agnostic learning but harder than realizable PAC learning. However, reliable agnostic learning is already hard for very basic classes. For example, reliable agnostic learning of conjunctions is as hard as PAC learning DNFs, a long-standing open problem widely believed to be computationally intractable.

The recent work of \cite{goel2024tolerantalgorithmslearningarbitrary} showed that, in distribution-specific settings where the \emph{training distribution} is fixed to be, for example, the standard Gaussian distribution or the uniform distribution on the Boolean hypercube, PQ learning is tractable for a broad range of classes. In particular, they showed that any class admitting low-degree $\calL_2$ sandwiching polynomials with respect to the training distribution can be efficiently learned in the PQ model. Their result captures several classes, especially in light of recent $\calL_2$ sandwiching degree bounds of \cite{klivans2026sandwiching}. Nevertheless, prior to our work, no subexponential-time PQ learning algorithm was known for DNFs when the training distribution is uniform over the Boolean hypercube. More broadly, constant-depth circuits are known to have quasipolynomial $\calL_1$ sandwiching degree with respect to the uniform training distribution \cite{bazzi2009polylogarithmic,razborov2009simple,10.1145/1754399.1754401,tal:LIPIcs.CCC.2017.15,hatami2023theory}, but the best known bound on their $\calL_2$ sandwiching degree is exponentially worse, namely polynomial in the circuit size \cite{odonnellservedio03ac0,klivans2024testablelearningdistributionshift}.

\paragraph{Sandwiching Polynomials.}
The notion of sandwiching polynomials originates in the pseudorandomness literature; see \cite{hatami2023theory} and references therein. In this context, most known results concern the uniform distribution over the Boolean hypercube \cite{bazzi2009polylogarithmic,razborov2009simple,10.1145/1754399.1754401,tal:LIPIcs.CCC.2017.15,diakonikolas2010bounded,diakonikolas2010ptf}, although sandwiching bounds were also shown for certain continuous distributions \cite{gopalan2010fooling,kanefoolingpfs,KKM13}.

Recent work in learning theory has highlighted the role of sandwiching approximation in challenging learning settings, including testable learning \cite{rubinfeld2022testingdistributionalassumptionslearning,gollakota2023moment,slot2024testably}, learning with distribution shift \cite{klivans2024testablelearningdistributionshift,goel2024tolerantalgorithmslearningarbitrary}, and learning under contamination \cite{klivans2024learningac0,klivans2026poweriterativefilteringsupervised}. This has motivated further work establishing sandwiching-degree bounds for broader concept classes, under broader families of distributions, and with higher-order error guarantees, namely $\calL_r$ sandwiching for $r>1$ \cite{klivans2026sandwiching}.

Nevertheless, substantial gaps remain between the best known $\calL_1$ and $\calL_2$ sandwiching-degree bounds for several basic classes. For constant-depth circuits under the uniform distribution, the best known bounds exhibit an exponential separation. For degree-$2$ polynomial threshold functions under the uniform distribution on the Boolean hypercube, and for constant-degree polynomial threshold functions under the Gaussian distribution, the $\calL_1$ sandwiching degree is polynomial in the inverse approximation error, whereas the best known $\calL_2$ sandwiching-degree bounds scale polynomially with the dimension $d$ \cite{klivans2026sandwiching}.

\paragraph{TDS Learning.}
The framework of TDS learning was introduced by \cite{klivans2024testablelearningdistributionshift}. PQ learning is stronger than TDS learning: in particular, it implies tolerant TDS learning, a stronger variant in which the learner is guaranteed to accept under sufficiently small distribution shift. Although PQ learning was defined before TDS learning, the first distribution-specific efficient algorithms for learning with distribution shift were developed in the TDS framework.

The first general result for TDS learning showed that $\calL_2$ sandwiching implies efficient TDS learning \cite{klivans2024testablelearningdistributionshift}. This was later strengthened by \cite{chandrasekaran2024efficient}, who showed that $\calL_1$ sandwiching already suffices. Our work shows that $\calL_1$ sandwiching also suffices for tolerant TDS learning and PQ learning. Some positive results for special concept classes in the TDS model have not yet been extended to PQ learning, or even to tolerant TDS learning \cite{klivans2024testablelearningdistributionshift,klivans2024learning,chandrasekaran2024efficient}. Nevertheless, existing TDS learning results have so far often admitted extensions to the PQ setting through suitable filtering procedures. Our ICF algorithm provides an additional technical tool for carrying out such extensions.

\paragraph{Chow Parameters.}
The notion of Chow parameters originates with Chow's characterization of Boolean threshold functions \cite{chow-original}. Since then, Chow parameters have played a role in several areas of theoretical computer science, including circuit complexity and threshold logic \cite{winder1971chow,goldmann1992majority}, game theory and the design of voting systems \cite{banzhaf1965weighted,dubey1979mathematical}, and learning theory \cite{bendavid1998learning,birkendorf1998restricted,chow-ltf,chandrasekaran2025fully}. A common use of Chow parameters is through structural theorems showing that certain Boolean-valued functions are determined, either exactly or robustly, by their low-degree Chow parameters \cite{chow-original,chow-ltf,chow-optimal,chow-ptf}. 

More recently, work in testable learning \cite{goel2025testingnoiseassumptions} and learning with distribution shift \cite{chandrasekaran2024efficient} has highlighted a different role for Chow-type expectations: not merely as parameters that characterize the underlying function, but as a tool for certifying the performance and reliability of learning algorithms. In this work, we show that such Chow-type expectations can also be used to obtain state-of-the-art guarantees in the PQ learning model.

\section{Preliminaries}

\paragraph{Notation.} Let \(\natural = \{0, 1, \dots\}\) denote the set of natural numbers and \(\natural^\ast = \natural \setminus \{0\}\). For any \(n \in \natural^\ast\), let \([n] = \{k \in \natural^\ast : k \leq n\}\). For any \(d, \ell \in \natural\), let \(\calP_{d,\ell}\) denote the set of \(d\)-variate polynomials, whose degree is at most \(\ell\). For any multiset \(S\), we let \(\calU(S)\) denote the uniform distribution on \(S\). For any real-valued function \(f\), we denote \(\bar f = 1 - f\).

\begin{definition}[Hypercontractivity] \label{def:hypercontractivity} Let any \(A \geq 1\). A distribution \(\calD\) on \(\real^d\) is said to be \(A\)-hypercontractive (with respect to polynomials) if for any polynomial \(p \colon \real^d \to \real\) and \(t \geq 2\):
\begin{enumerate}
    \item \(\ev{x \sim \calD}{\abs{p(x)}^t} \leq (At)^{t \deg(p)} \ev{x \sim \calD}{\abs{p(x)}}^t \) and
    \item the absolute value of any degree-\(1\) monomial has finite expectation under \(\calD\).
\end{enumerate}
\end{definition}

Hypercontractive distributions are a wide class of probability distributions, which
includes Gaussians, all log-concave distributions on \(\real^d\) (see \cite{bobkov2001some}), as well as product distributions over \(\{\pm 1\}^d\).

\section{Iterative Chow Filtering}

In this section, we present the ICF algorithm (\Cref{alg:icf}), which is the main tool for our results in PQ learning and tolerant TDS learning. The algorithm takes as input a classifier \(\hat f \colon \calX \to \{0, 1\}\), a sample \(S'\) from a test marginal \(\calD'\) and a reference sample from a training marginal \(\calD\), and outputs a selector \(s \colon \calX \to \{0, 1\}\). The goal is that the selector keeps test points whose low-degree polynomial behavior, restricted to the classifier-relevant regions specified by \(\calF = \{\hat f, 1- {\hat f}\}\), is consistent with \(\calD\), and rejects points that seem to be outliers. Specifically, ICF tries to enforce the following:
\begin{enumerate}
\item For each \(f \in \mathcal{F}\) and for every nonnegative low-degree polynomial \(p\) the following holds: If the expectation of \(f(x)p(x)\) under \(\calD\) is small, then the corresponding expectation under the part of \(\calD'\) selected by \(s\), should be small as well.
\item The probability that \(s\) rejects a sample drawn from \(\calD\) (i.e., \(\pr{x \sim \calD}{s(x) = 0}\)) is small.
\end{enumerate}
Concretely, we prove the following theorem. For the proof see \Cref{thm:icf_} in \Cref{appendix:icf}.

\begin{theorem}[Iterative Chow Filtering] \label{thm:icf}
        Let any \(\ell \in \natural^\ast\), \(A \geq 1\), \(R > 1\), \(\beta > 0\), \(\epsilon, \delta \in (0,1)\) and let \(\hat f \colon \real^d \to \{0, 1\}\). Let \(\calD\) be any \(A\)-hypercontractive distribution on \(\real^d\) and let \(\calD'\) be any distribution on \(\real^d\). Let \(S\) and \(S'\) be multisets of i.i.d. examples drawn from \(\calD\) and \(\calD'\) respectively. If \(\abs{S}, \abs{S'} \geq \mathrm{poly}(R, \beta, A^\ell, (d+1)^\ell, 1/\epsilon, \log(1/\delta)^\ell) \), then \Cref{alg:icf}, given \(\hat f, S, S', \ell, R, \beta, \epsilon\) as input, runs in time \(\mathrm{poly}\left(\abs{S}, \abs{S'}, (d+1)^\ell\right)\) and outputs a succinct polynomial-time-computable description of a function \(s \colon \real^d \to \{0, 1\}\) that satisfies the following guarantees.
    \begin{enumerate}
        \item Let any \(p, q \in \calP_{d, \ell}\) such that \(\max\{\ev{x \sim \calD}{p(x)^2}, \ev{x \sim \calD}{q(x)^2}\} \leq \beta\). Then, with probability at least \(1 - \delta\) over \(S\) and \(S'\), we have that
        \begin{align*}
            \ev{x \sim \calD'}{\hat f(x)p(x)s(x)} &\leq (R+\epsilon)(1+\epsilon)\ev{x \sim \calD}{\hat f(x)\abs{p(x)}} + 2\epsilon \, ,\\
            \ev{x \sim \calD'}{\bar {\hat f}(x)q(x)s(x)} &\leq (R+\epsilon)(1+\epsilon)\ev{x \sim \calD}{ \bar {\hat f}(x) \abs{q(x)}} + 2\epsilon \, .
        \end{align*}
        \item With probability at least \(1 - \delta\) over \(S\) and \(S'\), it holds
        \(
        \pr{x \sim \calD}{s(x)=0} \leq (1+\epsilon)/R
        \).
    \end{enumerate}
\end{theorem}

The hyperparameter \(R\) in \Cref{thm:icf} determines how selective the filtering is, that is, larger values of \(R\) imply that the rejection rate of \(s\) under \(\calD\) will be smaller. The cost of larger choices of \(R\) is that the first (Chow-transfer) guarantee of \Cref{thm:icf} holds with a worse (larger) multiplicative factor.

\begin{algorithm}[htbp]
\caption{Iterative Chow Filtering}\label{alg:icf}
\begin{algorithmic}[1]
\Require Classifier $\hat{f}$, sample \(S,S'\) from \(\calD,\calD'\) respectively, maximum polynomial degree \(\ell\), multiplicative slack \(R>1\), polynomial variance bound \(\beta > 0\), additive error \(\epsilon \in (0, 1)\).
\Ensure A succinct description of a selector function \(s \colon \real^d \to \{0, 1\}\).
\State Let \(B \gets 2\sqrt{\frac{2R(d+1)^\ell \beta}{ \epsilon}} \); \(\Delta \gets \frac{\epsilon^2}{B(2R+\epsilon)}\); \(\calF \gets \{\hat f, \bar{\hat f}\}\).
\State For \(f \in \calF\), define \(\calP(f) = \{p \in \calP_{d, \ell} : \ev{x \sim \calU(S)}{p(x)
^2} \leq 2\beta , \ev{x \sim \calU(S)}{f(x)\abs{p(x)} } \leq \frac{2\epsilon}{2R+\epsilon} \}\).
\State Let \(S_0 \gets \{x \in S' : \max_{f \in \calF}\max_{p \in \calP(f)} f(x)\abs{p(x)} \leq B \}\). \label{line:icf:S0}

\For{\(i \in [ \abs{S_0} + 1 ] \)}
\State Let \((f_i^\star, p_i^\star)\) be a solution and \(\mu_i^\star\) the optimal value of the following optimization problem: \label{line:icf:obj}
\[
\max_{f \in \calF} \max_{p \in \calP(f)} \dfrac{1}{\abs{S'}} \sum_{x \in S_{i-1}}f(x)p(x)
\] 
\If{\(\mu_i^\star \leq \epsilon\)} \label{line:icf:terminate}
\Return \(s \colon \real^d \to \{0, 1\}\) such that for all \(x \in \real^d\), it holds \label{line:icf:s}
\[s(x) = \INDICATOR{\max_{f \in \calF}\max_{p \in \calP(f)} f(x)\abs{p(x)} \leq B} \prod_{j=1}^{i-1}\indicator{f_j^\star(x) \abs{p_j^\star(x)} \leq \tau_j^\star} \, .
\]
\Else
\State Let \(\tau_i^\star\) be the minimum nonnegative real number such that
\[\dfrac{\abs{S_{i-1}}}{\abs{S'}} \pr{x \sim \calU(S_{i-1})}{f_i^\star(x)\abs{p_i^\star(x)} > \tau_i^\star} \geq R\pr{x \sim \calU(S)}{f_i^\star(x)\abs{p_i^\star(x)} > \tau_i^\star} + \Delta \, . \] \label{line:icf:tau}
\State \(S_i \gets S_{i-1} \setminus \{x \in S_{i-1} : f_i^\star(x)\abs{p_i^\star(x)} > \tau_i^\star \}\).\label{line:icf:Si} \Comment{Filtering step}
\EndIf
\EndFor
\end{algorithmic}
\end{algorithm}

ICF starts with a set of test points \(S'\) and performs a preliminary filtering step which ensures boundedness (see \cref{line:icf:S0}). It proceeds with iteratively solving the optimization problem of \cref{line:icf:obj} to find some \(f \in \calF\) and some low-degree polynomial \(p\) such that \(f(x) p(x)\) is large over the current set of accepted test points. Then, it uses \(f, p\) and the current set of accepted test points to find a rejection threshold \(\tau\) and discards test points such that \(f(x) \abs{p(x)} > \tau\) (see \cref{line:icf:Si}). After several filtering iterations, once the optimal value of the optimization problem in \cref{line:icf:obj} becomes sufficiently small, ICF outputs a selector, which is the intersection of all filtering rules accumulated over all iterations.

\paragraph{Proof Overview.} The main idea behind the analysis of \Cref{alg:icf} is that whenever the terminating condition of \cref{line:icf:terminate} does not hold, there exists some \(f \in \calF\) and some low-degree polynomial \(p\) such that the expectation \(f(x)\abs{p(x)}\) under the currently accepted part of \(\calD'\) is \(R\) times larger than the corresponding expectation under \(\calD\). We show that in that case, there exists some threshold \(\tau \geq 0\) such that \(\calD'\) assigns \(R\) times more mass on the currently accepted test points with \(f(x)\abs{p(x)} > \tau\) compared to \(\calD\). This implies that if we reject the test points with \(f(x)\abs{p(x)} > \tau\), then the mass of rejected points under \(\calD\) is at most \(1/R\) times the corresponding mass under \(\calD'\) (which implies that the rejection rate of \(s\) will be at most \(1/R\)). Finally, we show that in each iteration the algorithm removes at least a \(\Delta\)-fraction of the initial test sample, so the algorithm will eventually terminate.

\paragraph{Comparison with Prior Filtering Algorithms.}
ICF shares a similar filtering framework with the Iterative Polynomial Filtering
(IPF) algorithm of \cite{klivans2026poweriterativefilteringsupervised} and the outlier removal method of \cite{goel2024tolerantalgorithmslearningarbitrary}. The main difference is that ICF uses Chow-type expectations and handles localized quantities such as \(f(x)p(x)\) and \(f(x)|p(x)|\), whereas the algorithms of \cite{klivans2026poweriterativefilteringsupervised} and \cite{goel2024tolerantalgorithmslearningarbitrary} handle purely polynomial quantities that do not depend on a classifier $f$.

\section{PQ Learning}

In this section, we state our main result for PQ learning. Recall that the goal in PQ learning is to find a classifier \(\hat f \colon \calX \to \{0, 1\}\) and a selector \(s \colon \calX \to \{0, 1\}\) such that (i) \(\hat f\) achieves small error on the part of \(\Dtest\) selected by \(s\), and (ii) the rejection rate of \(s\) is small under \(\calD_\calX^\mathrm{train}\).

\paragraph{Algorithm Sketch.} The steps of our PQ learner are as follows (see also \Cref{thm:pq_}). 
\begin{itemize}
    \item First, we learn a classifier \(\hat f\)
via \(\calL_1\) polynomial regression on a training sample from $\Dtrain$.
\item Then, we run ICF with input $\hat f$ and examples from \(\calD_{\calX}^\mathrm{train}\) and \(\calD_{\calX}^\mathrm{test}\) to obtain a
selector \(s\).
\item Finally, we output $(\hat f,s)$.
\end{itemize}

Since \(\hat f\) is learned using labeled examples from the training distribution, we can easily guarantee that \(\hat f\) achieves low training error. The main difficulty in proving the correctness of our algorithm, lies in identifying a condition that will allow us to transfer the error from the training distribution to the selected part of the test distribution. To this end, we establish the following transfer lemma.

\begin{lemma}[Chow Transfer] \label{lemma:chow_transfer} Let \(\calX \subseteq \real^d\). Let \(\calD^\mathrm{train}, \calD^\mathrm{test}\) be distributions on \(\calX \times \{0, 1\}\) and let \(\calD\) and \(\calD'\) be their \(\calX\)-marginals respectively. Let \(\epsilon > 0\), \(R > 1\), let \(\calC \subseteq \{0, 1\}^\calX\), let any \(f^\star \in \calC\) such that \(\errtrain(f^\star) = \lambda_\mathrm{train}\) and \(\errtest(f^\star) = \lambda_\mathrm{test}\) and let \(p_\mathrm{up}, p_\mathrm{down}\) be \(\epsilon/(4R)\)-approximate \(\calL_1\) sandwiching polynomials for \(f^\star\) under \(\calD\).\footnote{Recall that this means that \(\pdown(x) \leq f^\star(x) \leq \pup(x)\) for all \(x \in \calX\) and \(\ev{x \sim \calD}{\pup(x) - \pdown(x)} \leq \epsilon/(4R)\).} Suppose that for some \(\hat f \colon \calX \to \{0, 1\}\) and \(s \colon \calX \to \{0, 1\}\), we have
\begin{enumerate}
    \item \(\ev{x \sim \calD'}{\hat f(x) \bar p_\mathrm{down}(x)s(x)} \leq R \ev{x \sim \calD}{\hat f(x) \bar p_\mathrm{down}(x) } + \epsilon/4\), and
    \item \(\ev{x \sim \calD'}{\bar {\hat f}(x) p_\mathrm{up}(x)s(x)} \leq R \ev{x \sim \calD}{\bar {\hat f}(x)p_\mathrm{up}(x)} + \epsilon/4\).
\end{enumerate}
    Then, it holds
    \(\pr{(x, y) \sim \calD^\mathrm{test}}{ \hat f(x) \neq y \land s(x)=1 } \leq \lambda_\mathrm{test} + R ( \lambda_\mathrm{train} + \errtrain(\hat f) ) + \epsilon
    \).
\end{lemma}

In words, \Cref{lemma:chow_transfer} says the following: If we find a selector \(s\) such that the Chow-type expectations of $\hat{f}$ and $\bar{\hat f}$ against appropriate choices of polynomials with respect to the selected part of \(\calD^\mathrm{test}\) are bounded by their analogues with respect to \(\calD^\mathrm{train}\) up to a multiplicative factor of \(R\), then the error of \(\hat f\) on the selected part of \(\calD^\mathrm{test}\) is roughly bounded by \(R(\lambda + \errtrain(\hat f))\).

By the first guarantee of \Cref{thm:icf}, a selector with the aforementioned property can be obtained via our ICF algorithm.\footnote{The first guarantee of \Cref{thm:icf} is applied on the polynomials \(p_\mathrm{up}\) and \(\bar p_\mathrm{down}\) of \Cref{lemma:chow_transfer}. Those polynomials only depend on \(\calD^\mathrm{train}\), \(\calD^\mathrm{test}\) and \(\calC\) and, although they are unknown to the algorithm, \Cref{thm:icf} can still be applied for them, since it holds for any fixed pair of low-degree polynomials.} Moreover, by the second guarantee of \Cref{thm:icf} any selector obtained via the ICF algorithm, ensures low rejection rate (roughly at most \(1/R\)) with respect to \(\calD^\mathrm{train}_\calX\). Hence, choosing \(R = 1/\eta\) would imply that the error of \(\hat f\) on the selected part of \(\Dtest\) is at most \((\lambda+\errtrain(\hat f))/\eta\) and the rejection rate of the selector is at most \(\eta\). This provides a high level idea of the correctness of our PQ learning algorithm. We now prove \Cref{lemma:chow_transfer}.

\begin{proof}[Proof of \Cref{lemma:chow_transfer}]
    Let \(\epsilon' = \epsilon/(4R)\). We have
    \begin{align*}
        \pr{x \sim \calD'}{ \hat f(x) \neq f^\star(x) \land s(x)=1 } &= \ev{x \sim \calD'}{\hat f(x) \bar{f}^\star(x) s(x)} + \ev{x \sim \calD'}{ \bar {\hat f}(x) f^\star(x) s(x)} \\
        &\leq \ev{x \sim \calD'}{\hat f(x) \bar p_\mathrm{down}(x)s(x)} + \ev{x \sim \calD'}{ \bar {\hat f}(x) p_\mathrm{up}(x) s(x)} \\
        &\leq R\left(\ev{x \sim \calD}{\hat f(x) \bar p_\mathrm{down}(x)} + \ev{x \sim \calD}{ \bar {\hat f}(x) p_\mathrm{up}(x)} \right) + \epsilon/2 \, ,
    \end{align*}
    where we used the fact that \(\pdown \leq f^\star \leq \pup\) and the assumption of \Cref{lemma:chow_transfer}. Now, we upper bound each of \(\ev{x \sim \calD}{\hat f(x) \bar p_\mathrm{down}(x)}\) and \(\ev{x \sim \calD}{ \bar {\hat f}(x) p_\mathrm{up}(x)}\) separately. We have
    \begin{align*}
    \ev{x \sim \calD}{\hat f(x) \bar p_\mathrm{down}(x)}& = \ev{x \sim \calD}{\hat f(x) \bar{f}^\star(x) } + \ev{x \sim \calD}{\hat f(x) (\bar p_\mathrm{down}(x) -\bar f^\star(x))} \\
    &\leq \ev{x \sim \calD}{\hat f(x) \bar{f}^\star(x) } + \ev{x \sim \calD}{p_\mathrm{up}(x)-p_\mathrm{down}(x)} 
    \leq \ev{x \sim \calD}{\hat f(x) \bar{f}^\star(x)} + \epsilon' \, ,
    \end{align*}
    where we used the fact that \(\bar p_\mathrm{down} -\bar f^\star \leq \pup - \pdown\) and \(\ev{x \sim \calD}{p_\mathrm{up}(x)-p_\mathrm{down}(x)} \leq \epsilon'\). Similarly,
    \[
    \ev{x \sim \calD}{\bar {\hat f}(x) p_\mathrm{up}(x)} = \ev{x \sim \calD}{\bar {\hat f}(x) f^\star(x)} + \ev{x \sim \calD}{\bar {\hat f}(x) (p_\mathrm{up}(x) -f^\star(x))}
    \leq \ev{x \sim \calD}{\bar {\hat f}(x) f^\star(x)} + \epsilon' \, .
    \]
    Putting everything together, we have
    \begin{align*}
    \pr{x \sim \calD'}{ \hat f(x) \neq f^\star(x) \land s(x)=1 } &\leq R \left( \ev{x \sim \calD}{\hat f(x) \bar{f}^\star(x) } + \ev{x \sim \calD}{\bar {\hat f}(x) f^\star(x)} + 2 \epsilon' \right) + \epsilon/2 \\
    &= R \left( \pr{x \sim \calD}{ \hat f(x) \neq f^\star(x) } + 2 \epsilon' \right) +\epsilon/2 \\
    &\leq R \left( \pr{(x,y) \sim \calD^\mathrm{train}}{ f^\star(x) \neq y } + \pr{(x,y) \sim \calD^\mathrm{train}}{ \hat f(x) \neq y } \right) +\epsilon \, ,
    \end{align*}
    where we applied the triangle inequality in the last line. Finally, another application of the triangle inequality implies that
    \begin{align*}
    \pr{(x,y) \sim \calD^\mathrm{test}}{ \hat f (x) \neq y \land s(x)=1 } &\leq \pr{(x,y) \sim \calD^\mathrm{test}}{f^\star(x) \neq y } + \pr{x \sim \calD'}{\hat f(x) \neq f^\star(x) \land s(x)=1 } \\
    &\leq \lambda_\mathrm{test} + R \left( \lambda_\mathrm{train} + \pr{(x, y) \sim \calD^\mathrm{train}}{ \hat f(x) \neq y } \right) + \epsilon \, ,
    \end{align*}
    which concludes the proof.
\end{proof}
We now state our main theorem for PQ learning. For the proof and applications to standard classes see \Cref{appendix:pq}.

\begin{theorem}[PQ Learning via \(\calL_1\) Sandwiching] \label{thm:pq} Let any \(\epsilon, \delta, \eta \in (0, 1)\) and \(A \geq 1\). Let \(\calX \subseteq \real^d \) and let \(\calD\) be any \(A\)-hypercontractive distribution on \(\calX\). Let \(\calC \subseteq \{0, 1\}^\calX\) be a concept class and let \(\ell=\calL_1(\calC, \calD, \epsilon\eta/8)\). Then, there exists an \((\epsilon, \delta, \eta)\)-PQ learner for \(\calC\) with respect to \(\calD\) with sample complexity and runtime \(\mathrm{poly} ((d+1)^\ell, A^\ell, 1/\epsilon, 1/\eta, \log(1/\delta)^\ell )\).
\end{theorem}

\section*{Limitations and Future Work}
Our results rely on the existence of $\mathcal{L}_1$ sandwiching polynomials with respect to a hypercontractive training distribution. Whether such sandwiching conditions are necessary for efficient algorithms in the PQ setting remains an open question. Moreover, although our error bounds are sharp and are tight in certain special cases, including the construction of \cite{goldwasser2020beyond}, it is unclear whether these bounds can be improved in other special cases, even by computationally inefficient algorithms.

\bibliography{references}

@inproceedings{rubinfeld2022testingdistributionalassumptionslearning,
author = {Rubinfeld, Ronitt and Vasilyan, Arsen},
title = {Testing Distributional Assumptions of Learning Algorithms},
year = {2023},
isbn = {9781450399135},
publisher = {Association for Computing Machinery},
address = {New York, NY, USA},
url = {https://doi.org/10.1145/3564246.3585117},
doi = {10.1145/3564246.3585117},
booktitle = {Proceedings of the 55th Annual ACM Symposium on Theory of Computing},
pages = {1643–1656},
numpages = {14},
location = {Orlando, FL, USA},
series = {STOC 2023}
}

@InProceedings{klivans2024testablelearningdistributionshift,
  title = 	 {Testable Learning with Distribution Shift},
  author =       {Klivans, Adam and Stavropoulos, Konstantinos and Vasilyan, Arsen},
  booktitle = 	 {Proceedings of Thirty Seventh Conference on Learning Theory},
  pages = 	 {2887--2943},
  year = 	 {2024},
  editor = 	 {Agrawal, Shipra and Roth, Aaron},
  volume = 	 {247},
  series = 	 {Proceedings of Machine Learning Research},
  month = 	 {30 Jun--03 Jul},
  publisher =    {PMLR},
  url = 	 {https://proceedings.mlr.press/v247/klivans24a.html}
}

@article{goldwasser2020beyond,
  title={Beyond perturbations: Learning guarantees with arbitrary adversarial test examples},
  author={Goldwasser, Shafi and Kalai, Adam Tauman and Kalai, Yael and Montasser, Omar},
  journal={Advances in Neural Information Processing Systems},
  volume={33},
  pages={15859--15870},
  year={2020}
}

@inproceedings{goel2024tolerantalgorithmslearningarbitrary,
 author = {Goel, Surbhi and Shetty, Abhishek and Stavropoulos, Konstantinos and Vasilyan, Arsen},
 booktitle = {Advances in Neural Information Processing Systems},
 doi = {10.52202/079017-3969},
 editor = {A. Globerson and L. Mackey and D. Belgrave and A. Fan and U. Paquet and J. Tomczak and C. Zhang},
 pages = {124979--125018},
 publisher = {Curran Associates, Inc.},
 title = {Tolerant Algorithms for Learning with Arbitrary Covariate Shift},
 url = {https://proceedings.neurips.cc/paper_files/paper/2024/file/e209210eae282e23e305df49fbb2769c-Paper-Conference.pdf},
 volume = {37},
 year = {2024}
}

@article{bobkov2001some,
  title={Some generalizations of Prokhorov's results on Khinchin-type inequalities for polynomials},
  author={Bobkov, Sergey G},
  journal={Theory of Probability \& Its Applications},
  volume={45},
  number={4},
  pages={644--647},
  year={2001},
  publisher={SIAM}
}

@INPROCEEDINGS{kanefoolingpfs,
  author={Kane, Daniel M.},
  booktitle={2011 IEEE 26th Annual Conference on Computational Complexity}, 
  title={k-Independent Gaussians Fool Polynomial Threshold Functions}, 
  year={2011},
  volume={},
  number={},
  pages={252-261},
  doi={10.1109/CCC.2011.13}}

@article{doi:10.1137/060649057,
author = {Kalai, Adam Tauman and Klivans, Adam R. and Mansour, Yishay and Servedio, Rocco A.},
title = {Agnostically Learning Halfspaces},
journal = {SIAM Journal on Computing},
volume = {37},
number = {6},
pages = {1777-1805},
year = {2008},
doi = {10.1137/060649057},

URL = { 
    
        https://doi.org/10.1137/060649057
    
    

},
eprint = { 
    
        https://doi.org/10.1137/060649057
    
    

}
}

@article{FERGER201496,
title = {Optimal constants in the Marcinkiewicz–Zygmund inequalities},
journal = {Statistics \& Probability Letters},
volume = {84},
pages = {96-101},
year = {2014},
issn = {0167-7152},
doi = {https://doi.org/10.1016/j.spl.2013.09.029},
url = {https://www.sciencedirect.com/science/article/pii/S0167715213003271},
author = {Dietmar Ferger}
}

@inproceedings{
chandrasekaran2025learning,
title={Learning Neural Networks with Distribution Shift: Efficiently Certifiable Guarantees},
author={Gautam Chandrasekaran and Adam Klivans and Lin Lin Lee and Konstantinos Stavropoulos},
booktitle={The Thirteenth International Conference on Learning Representations},
year={2025},
url={https://openreview.net/forum?id=ed7zI29lRF}
}

@article{Harsha_2018,
   title={On polynomial approximations to AC},
   volume={54},
   ISSN={1098-2418},
   url={http://dx.doi.org/10.1002/rsa.20786},
   DOI={10.1002/rsa.20786},
   number={2},
   journal={Random Structures \&amp; Algorithms},
   publisher={Wiley},
   author={Harsha, Prahladh and Srinivasan, Srikanth},
   year={2018},
   month=July, pages={289–303} }

@InProceedings{tal:LIPIcs.CCC.2017.15,
  author =	{Tal, Avishay},
  title =	{{Tight Bounds on the Fourier Spectrum of AC0}},
  booktitle =	{32nd Computational Complexity Conference (CCC 2017)},
  pages =	{15:1--15:31},
  series =	{Leibniz International Proceedings in Informatics (LIPIcs)},
  ISBN =	{978-3-95977-040-8},
  ISSN =	{1868-8969},
  year =	{2017},
  volume =	{79},
  editor =	{O'Donnell, Ryan},
  publisher =	{Schloss Dagstuhl -- Leibniz-Zentrum f{\"u}r Informatik},
  address =	{Dagstuhl, Germany},
  URL =		{https://drops.dagstuhl.de/entities/document/10.4230/LIPIcs.CCC.2017.15},
  URN =		{urn:nbn:de:0030-drops-75258},
  doi =		{10.4230/LIPIcs.CCC.2017.15}
}

@article{10.1145/1754399.1754401,
author = {Braverman, Mark},
title = {Polylogarithmic independence fools AC0 circuits},
year = {2008},
issue_date = {June 2010},
publisher = {Association for Computing Machinery},
address = {New York, NY, USA},
volume = {57},
number = {5},
issn = {0004-5411},
url = {https://doi.org/10.1145/1754399.1754401},
doi = {10.1145/1754399.1754401},
journal = {J. ACM},
month = jun,
articleno = {28},
numpages = {10}
}

@article{chandrasekaran2024efficient,
  title={Efficient discrepancy testing for learning with distribution shift},
  author={Chandrasekaran, Gautam and Klivans, Adam and Kontonis, Vasilis and Stavropoulos, Konstantinos and Vasilyan, Arsen},
  journal={Advances in Neural Information Processing Systems},
  volume={37},
  pages={137263--137308},
  year={2024}
}

@inproceedings{chow-original,
author = {Chow, C. K.},
title = {On the characterization of threshold functions},
year = {1961},
publisher = {IEEE Computer Society},
address = {USA},
url = {https://doi.org/10.1109/FOCS.1961.24},
doi = {10.1109/FOCS.1961.24},
booktitle = {Proceedings of the 2nd Annual Symposium on Switching Circuit Theory and Logical Design (SWCT 1961)},
pages = {34–38},
numpages = {5},
series = {FOCS '61}
}

@inproceedings{chow-ptf,
author = {Diakonikolas, Ilias and Kane, Daniel M.},
title = {Degree-d chow parameters robustly determine degree-d PTFs (and algorithmic applications)},
year = {2019},
isbn = {9781450367059},
publisher = {Association for Computing Machinery},
address = {New York, NY, USA},
url = {https://doi.org/10.1145/3313276.3316301},
doi = {10.1145/3313276.3316301},
booktitle = {Proceedings of the 51st Annual ACM SIGACT Symposium on Theory of Computing},
pages = {804–815},
numpages = {12},
location = {Phoenix, AZ, USA},
series = {STOC 2019}
}

@inproceedings{chow-ltf,
author = {O'Donnell, Ryan and Servedio, Rocco A.},
title = {The chow parameters problem},
year = {2008},
isbn = {9781605580470},
publisher = {Association for Computing Machinery},
address = {New York, NY, USA},
url = {https://doi.org/10.1145/1374376.1374450},
doi = {10.1145/1374376.1374450},
booktitle = {Proceedings of the Fortieth Annual ACM Symposium on Theory of Computing},
pages = {517–526},
numpages = {10},
location = {Victoria, British Columbia, Canada},
series = {STOC '08}
}

@InProceedings{kalai21acovariateshift,
  title = 	 {Efficient Learning with Arbitrary Covariate Shift},
  author =       {Kalai, Adam Tauman and Kanade, Varun},
  booktitle = 	 {Proceedings of the 32nd International Conference on Algorithmic Learning Theory},
  pages = 	 {850--864},
  year = 	 {2021},
  editor = 	 {Feldman, Vitaly and Ligett, Katrina and Sabato, Sivan},
  volume = 	 {132},
  series = 	 {Proceedings of Machine Learning Research},
  month = 	 {16--19 Mar},
  publisher =    {PMLR},
  url = 	 {https://proceedings.mlr.press/v132/kalai21a.html}
}

@article{KALAI2012reliable,
title = {Reliable agnostic learning},
journal = {Journal of Computer and System Sciences},
volume = {78},
number = {5},
pages = {1481-1495},
year = {2012},
note = {JCSS Special Issue: Cloud Computing 2011},
issn = {0022-0000},
doi = {https://doi.org/10.1016/j.jcss.2011.12.026},
url = {https://www.sciencedirect.com/science/article/pii/S0022000012000256},
author = {Adam Tauman Kalai and Varun Kanade and Yishay Mansour}
}

@article{klivans2026sandwiching,
  title={Sandwiching Polynomials for Geometric Concepts with Low Intrinsic Dimension},
  author={Klivans, Adam R and Stavropoulos, Konstantinos and Vasilyan, Arsen},
  journal={arXiv preprint arXiv:2602.24178},
  year={2026}
}

@inproceedings{odonnellservedio03ac0,
author = {O'Donnell, Ryan and Servedio, Rocco A.},
title = {New degree bounds for polynomial threshold functions},
year = {2003},
isbn = {1581136749},
publisher = {Association for Computing Machinery},
address = {New York, NY, USA},
url = {https://doi.org/10.1145/780542.780592},
doi = {10.1145/780542.780592},
booktitle = {Proceedings of the Thirty-Fifth Annual ACM Symposium on Theory of Computing},
pages = {325–334},
numpages = {10},
location = {San Diego, CA, USA},
series = {STOC '03}
}

@article{bazzi2009polylogarithmic,
	author = {Bazzi, Louay MJ},
	journal = {SIAM Journal on Computing},
	number = {6},
	pages = {2220--2272},
	publisher = {SIAM},
	title = {Polylogarithmic independence can fool DNF formulas},
	volume = {38},
	year = {2009}}

@article{razborov2009simple,
  title={A simple proof of Bazzi’s theorem},
  author={Razborov, Alexander},
  journal={ACM Transactions on Computation Theory (TOCT)},
  volume={1},
  number={1},
  pages={1--5},
  year={2009},
  publisher={ACM New York, NY, USA}
}

@inproceedings{hatami2023theory,
  title={Theory of unconditional pseudorandom generators},
  author={Hatami, Pooya and Hoza, William},
  booktitle={Electron. Colloquium Comput. Complex., TR23-019},
  volume={1},
  number={1},
  pages={1},
  year={2023}
}

@inproceedings{
klivans2026poweriterativefilteringsupervised,
title={The Power of Iterative Filtering for Supervised Learning with (Heavy) Contamination},
author={Adam Klivans and Konstantinos Stavropoulos and Kevin Tian and Arsen Vasilyan},
booktitle={The Thirty-ninth Annual Conference on Neural Information Processing Systems},
year={2026},
url={https://openreview.net/forum?id=E7knuYAvpt}
}

@article{klivans2009learning,
  title={Learning Halfspaces with Malicious Noise.},
  author={Klivans, Adam R and Long, Philip M and Servedio, Rocco A},
  journal={Journal of Machine Learning Research},
  volume={10},
  number={12},
  year={2009}
}

@inproceedings{diakonikolas2018learning,
  title={Learning geometric concepts with nasty noise},
  author={Diakonikolas, Ilias and Kane, Daniel M and Stewart, Alistair},
  booktitle={Proceedings of the 50th Annual ACM SIGACT Symposium on Theory of Computing},
  pages={1061--1073},
  year={2018}
}

@article{diakonikolas2019robust,
  title={Robust estimators in high-dimensions without the computational intractability},
  author={Diakonikolas, Ilias and Kamath, Gautam and Kane, Daniel and Li, Jerry and Moitra, Ankur and Stewart, Alistair},
  journal={SIAM Journal on Computing},
  volume={48},
  number={2},
  pages={742--864},
  year={2019},
  publisher={SIAM}
}

@inproceedings{diakonikolas2019sever,
  title={Sever: A robust meta-algorithm for stochastic optimization},
  author={Diakonikolas, Ilias and Kamath, Gautam and Kane, Daniel and Li, Jerry and Steinhardt, Jacob and Stewart, Alistair},
  booktitle={International Conference on Machine Learning},
  pages={1596--1606},
  year={2019},
  organization={PMLR}
}

@InProceedings{klivans2024learningac0,
  title = 	 {Learning Constant-Depth Circuits in Malicious Noise Models},
  author =       {Klivans, Adam and Stavropoulos, Konstantinos and Vasilyan, Arsen},
  booktitle = 	 {Proceedings of Thirty Eighth Conference on Learning Theory},
  pages = 	 {3253--3263},
  year = 	 {2025},
  editor = 	 {Haghtalab, Nika and Moitra, Ankur},
  volume = 	 {291},
  series = 	 {Proceedings of Machine Learning Research},
  month = 	 {30 Jun--04 Jul},
  publisher =    {PMLR},
  url = 	 {https://proceedings.mlr.press/v291/klivans25a.html}
}

@article{diakonikolas2010bounded,
	author = {Diakonikolas, Ilias and Gopalan, Parikshit and Jaiswal, Ragesh and Servedio, Rocco A and Viola, Emanuele},
	journal = {SIAM Journal on Computing},
	number = {8},
	pages = {3441--3462},
	publisher = {SIAM},
	title = {Bounded independence fools halfspaces},
	volume = {39},
	year = {2010}}

@inproceedings{KKM13,
  author       = {Daniel M. Kane and
                  Adam R. Klivans and
                  Raghu Meka},
  editor       = {Shai Shalev{-}Shwartz and
                  Ingo Steinwart},
  title        = {Learning Halfspaces Under Log-Concave Densities: Polynomial Approximations
                  and Moment Matching},
  booktitle    = {{COLT} 2013 - The 26th Annual Conference on Learning Theory, June
                  12-14, 2013, Princeton University, NJ, {USA}},
  series       = {{JMLR} Workshop and Conference Proceedings},
  volume       = {30},
  pages        = {522--545},
  publisher    = {JMLR.org},
  year         = {2013},
  url          = {http://proceedings.mlr.press/v30/Kane13.html},
  bibsource    = {dblp computer science bibliography, https://dblp.org}
}

@inproceedings{gopalan2010fooling,
  title={Fooling functions of halfspaces under product distributions},
  author={Gopalan, Parikshit and O'Donnell, Ryan and Wu, Yi and Zuckerman, David},
  booktitle={2010 IEEE 25th Annual Conference on Computational Complexity},
  pages={223--234},
  year={2010},
  organization={IEEE}
}

@inproceedings{diakonikolas2010ptf,
  title={Bounded independence fools degree-2 threshold functions},
  author={Diakonikolas, Ilias and Kane, Daniel M and Nelson, Jelani},
  booktitle={2010 IEEE 51st Annual Symposium on Foundations of Computer Science},
  pages={11--20},
  year={2010},
  organization={IEEE}
}

@inproceedings{gollakota2023moment,
author = {Gollakota, Aravind and Klivans, Adam R. and Kothari, Pravesh K.},
title = {A Moment-Matching Approach to Testable Learning and a New Characterization of Rademacher Complexity},
year = {2023},
isbn = {9781450399135},
publisher = {Association for Computing Machinery},
address = {New York, NY, USA},
url = {https://doi.org/10.1145/3564246.3585206},
doi = {10.1145/3564246.3585206},
booktitle = {Proceedings of the 55th Annual ACM Symposium on Theory of Computing},
pages = {1657–1670},
numpages = {14},
location = {Orlando, FL, USA},
series = {STOC 2023}
}

@article{slot2024testably,
  title={Testably learning polynomial threshold functions},
  author={Slot, Lucas and Tiegel, Stefan and Wiedmer, Manuel},
  journal={Advances in Neural Information Processing Systems},
  volume={37},
  pages={3781--3831},
  year={2024}
}

@article{klivans2024learning,
  title={Learning Intersections of Halfspaces with Distribution Shift: Improved Algorithms and SQ Lower Bounds},
  author={Klivans, Adam R and Stavropoulos, Konstantinos and Vasilyan, Arsen},
  journal={The Thirty Seventh Annual Conference on Learning Theory},
  year={2024},
  organization={PMLR}
}

@article{winder1971chow,
  author  = {Winder, R. O.},
  title   = {Chow Parameters in Threshold Logic},
  journal = {Journal of the ACM},
  volume  = {18},
  pages   = {265--289},
  year    = {1971}
}

@article{goldmann1992majority,
  author  = {Goldmann, Mikael and H{\aa}stad, Johan and Razborov, Alexander},
  title   = {Majority Gates vs. General Weighted Threshold Gates},
  journal = {Computational Complexity},
  volume  = {2},
  pages   = {277--300},
  year    = {1992}
}

@article{banzhaf1965weighted,
  author  = {Banzhaf, John F.},
  title   = {Weighted Voting Doesn't Work: A Mathematical Analysis},
  journal = {Rutgers Law Review},
  volume  = {19},
  pages   = {317--343},
  year    = {1965}
}

@article{dubey1979mathematical,
  author  = {Dubey, Pradeep and Shapley, Lloyd S.},
  title   = {Mathematical Properties of the Banzhaf Power Index},
  journal = {Mathematics of Operations Research},
  volume  = {4},
  pages   = {99--131},
  year    = {1979}
}

@article{bendavid1998learning,
  author  = {Ben-David, Shai and Dichterman, Eli},
  title   = {Learning with Restricted Focus of Attention},
  journal = {Journal of Computer and System Sciences},
  volume  = {56},
  pages   = {277--298},
  year    = {1998}
}

@article{birkendorf1998restricted,
  author  = {Birkendorf, Andreas and Dichterman, Eli and Jackson, Jeffrey and Klasner, Nathaniel and Simon, Hans Ulrich},
  title   = {On Restricted-Focus-of-Attention Learnability of Boolean Functions},
  journal = {Machine Learning},
  volume  = {30},
  pages   = {89--123},
  year    = {1998}
}

@article{goel2025testingnoiseassumptions,
  title={Testing noise assumptions of learning algorithms},
  author={Goel, Surbhi and Klivans, Adam R and Stavropoulos, Konstantinos and Vasilyan, Arsen},
  journal={arXiv preprint arXiv:2501.09189},
  year={2025}
}

@article{chandrasekaran2025fully,
  title={A Fully Polynomial-Time Algorithm for Robustly Learning Halfspaces over the Hypercube},
  author={Chandrasekaran, Gautam and Klivans, Adam R and Stavropoulos, Konstantinos and Vasilyan, Arsen},
  journal={arXiv preprint arXiv:2511.07244},
  year={2025}
}

@article{chow-optimal,
author = {De, Anindya and Diakonikolas, Ilias and Feldman, Vitaly and Servedio, Rocco A.},
title = {Nearly Optimal Solutions for the Chow Parameters Problem and Low-Weight Approximation of Halfspaces},
year = {2014},
issue_date = {April 2014},
publisher = {Association for Computing Machinery},
address = {New York, NY, USA},
volume = {61},
number = {2},
issn = {0004-5411},
url = {https://doi.org/10.1145/2590772},
doi = {10.1145/2590772},
journal = {J. ACM},
month = apr,
articleno = {11},
numpages = {36}
}

@article{ben2006analysis,
  title={Analysis of representations for domain adaptation},
  author={Ben-David, Shai and Blitzer, John and Crammer, Koby and Pereira, Fernando},
  journal={Advances in neural information processing systems},
  volume={19},
  year={2006}
}

@article{blitzer2007learning,
  title={Learning bounds for domain adaptation},
  author={Blitzer, John and Crammer, Koby and Kulesza, Alex and Pereira, Fernando and Wortman, Jennifer},
  journal={Advances in neural information processing systems},
  volume={20},
  year={2007}
}

@article{article,
author = {Vaart, Aad and Wellner, Jon},
year = {2009},
month = {01},
pages = {103-107},
title = {A note on bounds for VC dimensions},
volume = {5},
isbn = {978-0-940600-78-2},
journal = {Institute of Mathematical Statistics collections},
doi = {10.1214/09-IMSCOLL508}
}

@article{lmnac0,
author = {Linial, Nathan and Mansour, Yishay and Nisan, Noam},
title = {Constant depth circuits, Fourier transform, and learnability},
year = {1993},
issue_date = {July 1993},
publisher = {Association for Computing Machinery},
address = {New York, NY, USA},
volume = {40},
number = {3},
issn = {0004-5411},
url = {https://doi.org/10.1145/174130.174138},
doi = {10.1145/174130.174138},
journal = {J. ACM},
month = jul,
pages = {607–620},
numpages = {14}
}

@inproceedings{mansour2009domadapt,
title	= {Domain Adaptation: Learning Bounds and Algorithms},
author	= {Yishay Mansour and Mehryar Mohri and Afshin Rostamizadeh},
year	= {2009},
URL	= {http://www.cs.nyu.edu/~mohri/postscript/nadap.pdf},
booktitle	= {Proceedings of The 22nd Annual Conference on Learning Theory (COLT 2009)},
address	= {Montr\'eal, Canada}
}

@inproceedings{david2010impossibility,
  title={Impossibility theorems for domain adaptation},
  author={David, Shai Ben and Lu, Tyler and Luu, Teresa and P{\'a}l, D{\'a}vid},
  booktitle={Proceedings of the Thirteenth International Conference on Artificial Intelligence and Statistics},
  pages={129--136},
  year={2010},
  organization={JMLR Workshop and Conference Proceedings}
}

@article{ben2010theory,
  title={A theory of learning from different domains},
  author={Ben-David, Shai and Blitzer, John and Crammer, Koby and Kulesza, Alex and Pereira, Fernando and Vaughan, Jennifer Wortman},
  journal={Machine learning},
  volume={79},
  pages={151--175},
  year={2010},
  publisher={Springer}
}

\newpage

\appendix

\crefalias{section}{appendix}
\crefalias{subsection}{appendix}

\section{Extended Preliminaries}

\subsection{Notation}

Let \(\natural = \{0, 1, \dots\}\) denote the set of natural numbers and \(\natural^\ast = \natural \setminus \{0\}\). For any \(n \in \natural^\ast\), let \([n] = \{k \in \natural^\ast : k \leq n\}\). Let \(\integer\) denote the set of integers. For any \(d, \ell \in \natural\), let \(\calP_{d,\ell}\) denote the set of \(d\)-variate polynomials, whose degree is at most \(\ell\). For any polynomial \(p\), let \(\deg(p)\) denote its degree. For any \(d \in \natural^\ast\), let \(\calN_d\) denote the standard normal distribution on \(\real^d\). For any multiset \(S\), we let \(\calU(S)\) denote the uniform distribution on \(S\). For any real-valued function \(f\), we denote \(\bar f = 1 - f\). For any real-valued functions \(f, g\), we use \(f \lesssim g\) (resp. \(f \gtrsim g\)) to indicate that there exists some universal constant \(C > 0\) such that \(f \leq Cg\) (resp. \(Cf \geq g\)). For any pair of distributions \(\calD, \calD'\) on some set, we let \(d_\mathrm{TV}(\calD, \calD')\) denote the total variation distance between \(\calD\) and \(\calD'\). For any pair of training and test distributions \(\Dtrain\) and \(\Dtest\) on \(\calX \times \calY\), we denote their marginals on \(\calX\) by \(\Dtrain_\calX\) and \(\Dtest_\calX\) respectively. Finally, for a given concept class \(\calC \subseteq \calY^\calX\), we denote \(\mathsf{opt}_\mathrm{train} = \inf_{f \in \calC} \pr{(x, y) \sim \Dtrain}{f(x) \neq y}\).

\subsection{Learning Theory}

\begin{definition}
    A halfspace in \(\real^d\) is any function \(h \colon \real^d \to \{\pm 1\}\) of the form \(h(x) = \mathrm{sign}(w \cdot x + b)\) for some unit vector \(w \in \real^d\) and \(b \in \real\).
\end{definition}

\begin{definition}
    Let any \(\ell \in \natural\). A degree-\(\ell\) polynomial threshold function (PTF) in \(\real^d\) is any function \(h \colon \real^d \to \{\pm 1\}\) of the form \(h(x) = \mathrm{sign}(p(x))\) for some \(p \in \calP_{d, \ell}\).
\end{definition}

\begin{definition}
    Let any set \(\calX\), \(r \in \natural^\ast\), \(\calH_1, \dots, \calH_r \subseteq \{0, 1\}^\calX\) and let \(V = \max_{i \in [r]} \mathrm{VCdim}(\calH_i)\). The OR of \(\calH_1, \dots, \calH_r\) is defined as \(\bigvee_{i=1}^r \calH_i = \{h_1 \lor \cdots \lor h_r : h_1 \in \calH_1, \dots, h_r \in \calH_r\}\) and the \textit{AND of \(\calH_1, \dots, \calH_r\)} is defined as \(\bigwedge_{i=1}^r \calH_i = \{h_1 \land \cdots \land h_r : h_1 \in \calH_1, \dots, h_r \in \calH_r\}\).
\end{definition}

\begin{fact} \label{fact:vc_union}
    Let any set \(\calX\), \(r \in \natural^\ast\), \(\calH_1, \dots, \calH_r \subseteq \{0, 1\}^\calX\) and let \(V = \max_{i \in [r]} \mathrm{VCdim}(\calH_i)\). Then, it holds \(\mathrm{VCdim}(\bigcup_{i=1}^r \calH_i) \lesssim V \log(V) + \log(r)\).
\end{fact}

\begin{fact}[\cite{article} and references therein] \label{fact:vc_and}
    Let any set \(\calX\), \(r \in \natural^\ast\), \(\calH_1, \dots, \calH_r \subseteq \{0, 1\}^\calX\) and let \(V = \max_{i \in [r]} \mathrm{VCdim}(\calH_i)\). Then, it holds \(\mathrm{VCdim}(\bigvee_{i=1}^r \calH_i) \lesssim V r \log(r)\) and \(\mathrm{VCdim}(\bigwedge_{i=1}^r \calH_i) \lesssim V r \log(r)\).
\end{fact}

\begin{fact} \label{fact:vc_ptf}
    For any \(d, \ell \in \natural^\ast\), the VC dimension of polynomial threshold functions of degree at most \(\ell\) in \(\real^d\) is at most \((d+1)^\ell\).
\end{fact}

\begin{fact} \label{fact:uc0}
    Let any set \(\calX\), any distribution \(\calD\) on \(\calX\) and any class \(\calH \subseteq \{0,1\}^\calX\) such that \(\mathrm{VCdim}(\calH) < \infty\). Let \(S\) be a multiset of \(n\) i.i.d. samples drawn from \(\calD\). Then, there exists some universal constant \(C > 0\) such that for any \(\delta \in (0, 1)\), with probability at least \(1 - \delta\) over \(S\), it holds
    \[
    \sup_{h \in \calH} \ABS{ \ev{x \sim \calU(S)}{h(x)} - \ev{x \sim \calD}{h(x)} } \leq C\sqrt{ \dfrac{\mathrm{VCdim}(\calH) \log(en) + \log(1/\delta)}{n}} \, .
    \]
\end{fact}

\subsection{Sandwiching Approximators}\label{appendix:sandwiching-bounds}

In this section, we present known literature results on the existence of \(\calL_1\) sandwiching polynomials for various concept classes.

\begin{fact}[\cite{10.1145/1754399.1754401, tal:LIPIcs.CCC.2017.15, Harsha_2018}] \label{fact:l1_degree_AC0}
    Let any \(\epsilon \in (0, 1)\). Let \(\calC\) be the class of depth-\(t\), size-\(s\) \(\mathsf{AC}^0\) circuits in \(\{\pm 1\}^d\). Then, \( \calL_1(\calC, \calU(\{\pm1\}^d), \epsilon) \leq O(\log(s))^{O(t)}\log(1/\epsilon)\).
\end{fact}

\begin{fact}[\cite{gopalan2010fooling}] \label{fact:l1_degree_dt}
    Let any \(\epsilon \in (0, 1)\). Let \(\calC\) be the class of depth-\(t\), size-\(s\) decision trees of halfspaces in \(\real^d\). Then, \( \calL_1(\calC, \calN_d, \epsilon) \leq \tilde O(t^4s^2/\epsilon^2)\) and \( \calL_1(\calC, \calU(\{\pm1\}^d), \epsilon) \leq \tilde O(t^4s^2/\epsilon^2)\).
\end{fact}

\begin{fact}[\cite{diakonikolas2010ptf}] \label{fact:l1_degree_degree2ptf}
    Let any \(\epsilon \in (0, 1)\). Let \(\calC\) be the class of degree-\(2\) PTFs in \(\real^d\). Then, \( \calL_1(\calC, \calN_d, \epsilon) \leq \tilde O(1/\epsilon^8)\) and \( \calL_1(\calC, \calU(\{\pm1\}^d), \epsilon) \leq \tilde O(1/\epsilon^9)\).
\end{fact}

\begin{fact}[\cite{kanefoolingpfs,slot2024testably}] \label{fact:l1_degree_degreekptf}
    Let any \(\epsilon \in (0, 1)\). Let \(\calC\) be the class of degree-\(k\) PTFs in \(\real^d\). Then, \( \calL_1(\calC, \calN_d, \epsilon) \leq O_k(\epsilon^{-4k \cdot 7^k})\).\footnote{Here \(O_k(\cdot)\) is hiding a multiplicative factor that scales as an arbitrary function of \(k\).}
\end{fact}

\begin{fact}[\cite{KKM13}] \label{fact:l1_degree_function_of_ltfs}
    Let any \(\epsilon \in (0, 1)\). Let \(\calC\) be the class of arbitrary functions of \(k\) halfspaces in \(\real^d\) and let \(\calD\) be any log-concave distribution on \(\real^d\). Then, \( \calL_1(\calC, \calD, \epsilon) \leq \exp((\log(\log(k)/\epsilon))^{O(k)}/\epsilon^4)\).
\end{fact}

\section{Iterative Chow Filtering} \label{appendix:icf}

In this section, we provide the proof of our main ICF theorem. In fact, here we state \Cref{alg:icf_} and \Cref{thm:icf_}, which are the more general versions of \Cref{alg:icf} and \Cref{thm:icf} respectively.

\begin{algorithm}[H]
\caption{Iterative Chow Filtering}\label{alg:icf_}
\begin{algorithmic}[1]
\Require Family of classifiers $\calF \subset \{\real^d \to \{0, 1\}\}$, sample \(S,S'\) from \(\calD,\calD'\) respectively, maximum polynomial degree \(\ell\), multiplicative slack \(R>1\), polynomial variance bound \(\beta > 0\), additive error \(\epsilon \in (0, 1)\).
\Ensure A succinct description of a selector function \(s \colon \real^d \to \{0, 1\}\).
\State Let \(B \gets 2\sqrt{\frac{2R(d+1)^\ell \beta}{ \epsilon}} \); \(\Delta \gets \frac{\epsilon^2}{B(2R+\epsilon)}\); \(t \gets 0\).
\State For \(f \in \calF\), define \(\calP(f) = \{p \in \calP_{d, \ell} : \ev{x \sim \calU(S)}{p(x)
^2} \leq 2\beta , \ev{x \sim \calU(S)}{f(x)\abs{p(x)} } \leq \frac{2\epsilon}{2R+\epsilon} \}\).
\State Let \(S_0 \gets \{x \in S' : \max_{f \in \calF}\max_{p \in \calP(f)} f(x)\abs{p(x)} \leq B \}\). \label{line:icf_:S0}

\For{\(i \in [ \abs{S_0} + 1 ] \)}
\State Let \((f_i^\star, p_i^\star)\) be a solution and \(\mu_i^\star\) the optimal value of the following optimization problem: \label{line:icf_:obj}
\[
\max_{f \in \calF} \max_{p \in \calP(f)} \dfrac{1}{\abs{S'}} \sum_{x \in S_{i-1}}f(x)p(x)
\]
\If{\(\mu_i^\star \leq \epsilon\)} \label{line:icf_:terminate}
\State Set \(t \gets i - 1\) and \textbf{break} the loop.
\Else
\State Let \(\tau_i^\star\) be the minimum nonnegative real number such that
\[\dfrac{\abs{S_{i-1}}}{\abs{S'}} \pr{x \sim \calU(S_{i-1})}{f_i^\star(x)\abs{p_i^\star(x)} > \tau_i^\star} \geq R\pr{x \sim \calU(S)}{f_i^\star(x)\abs{p_i^\star(x)} > \tau_i^\star} + \Delta \, . \] \label{line:icf_:tau}
\State \(S_i \gets S_{i-1} \setminus \{x \in S_{i-1} : f_i^\star(x)\abs{p_i^\star(x)} > \tau_i^\star \}\). \label{line:icf_:Si} \Comment{Filtering step}
\EndIf
\EndFor
\State \Return \(s \colon \real^d \to \{0, 1\}\) such that for all \(x \in \real^d\), it holds
\[s(x) = \INDICATOR{\max_{f \in \calF}\max_{p \in \calP(f)} f(x)\abs{p(x)} \leq B} \prod_{i=1}^{t}\indicator{f_i^\star(x) \abs{p_i^\star(x)} \leq \tau_i^\star} \, .
\]
\end{algorithmic}
\end{algorithm}

The difference of \Cref{thm:icf_} from \Cref{thm:icf} is that the guarantee of \Cref{thm:icf_} holds for any finite collection of classifiers \(\calF \subset \{\real^d \to \{0, 1\}\}\) instead of \(\calF = \{\hat f, \bar{\hat f}\}\). Moreover, \Cref{thm:icf_} provides an additional upper bound on the rejection rate of the selector, which is crucial to obtaining our tolerant TDS learning result (see \Cref{appendix:tolerant_tds}). Here, we will prove \Cref{thm:icf_} (which \Cref{thm:icf} follows trivially from).

We first prove the following proposition, which ensures that all steps of \Cref{alg:icf_} are well-defined.

\begin{proposition}
    In the setting of \Cref{alg:icf_}, for any iteration \(i\) such that \(\mu_i^\star > \epsilon\), there exists some \(\tau \in [0, B]\) that satisfies the condition of \cref{line:icf_:tau}.
\end{proposition}

\begin{proof}
    Let any iteration \(i\) such that \(\mu_i^\star > \epsilon\). Suppose, for contradiction, that for all \(\tau \in [0, B]\), it holds
    \[
    \dfrac{\abs{S_{i-1}}}{\abs{S'}} \pr{x \sim \calU(S_{i-1})}{f_i^\star(x)\abs{p_i^\star(x)} > \tau} < R\pr{x \sim \calU(S)}{f_i^\star(x)\abs{p_i^\star(x)} > \tau} + \Delta \, .
    \]
    Now, we may integrate over \(\tau \in [0, B]\) both sides of the above inequality, since the corresponding functions of \(\tau\) have a finite number of discontinuities.
    \[
    \dfrac{\abs{S_{i-1}}}{\abs{S'}} \int_{\tau=0}^B \pr{x \sim \calU(S_{i-1})}{f_i^\star(x)\abs{p_i^\star(x)} > \tau} \, \mathrm d \tau < R \int_{\tau=0}^B \pr{x \sim \calU(S)}{f_i^\star(x)\abs{p_i^\star(x)} > \tau} \, \mathrm d \tau + \Delta B \, .
    \]
    Using the tail-sum formula for nonnegative random variables, we get that
    \begin{align*}
    \ev{x \sim \calU(S_{i-1})}{f_i^\star(x) \abs{p_i^\star(x)}} &= \int_{\tau=0}^\infty \pr{x \sim \calU(S_{i-1})}{f_i^\star(x)\abs{p_i^\star(x)} > \tau} \, \mathrm d \tau \\
    &= \int_{\tau=0}^B \pr{x \sim \calU(S_{i-1})}{f_i^\star(x)\abs{p_i^\star(x)} > \tau} \, \mathrm d \tau \, ,    
    \end{align*}
    where the second equality follows from the fact that \(p_i^\star \in \calP(f_i^\star)\), \(S_{i-1} \subseteq S_0\) and \(f(x)\abs{p(x)} \leq B\) for any \(f\in \calF\), \(p \in \calP(f)\) and \(x \in S_0\) due to \cref{line:icf_:S0}. Similarly, we have
    \begin{align*}
    \ev{x \sim \calU(S)}{f_i^\star(x) \abs{p_i^\star(x)}} &= \int_{\tau=0}^\infty \pr{x \sim \calU(S)}{f_i^\star(x)\abs{p_i^\star(x)} > \tau} \, \mathrm d \tau \\
    &\geq \int_{\tau=0}^B \pr{x \sim \calU(S)}{f_i^\star(x)\abs{p_i^\star(x)} > \tau} \, \mathrm d \tau \, ,   
    \end{align*}
    where the inequality follows from the fact that the integrated function is nonnegative. Putting everything together, we have
    \begin{align*}
    \mu_i^\star &= \dfrac{\abs{S_{i-1}}}{\abs{S'}} \ev{x \sim \calU(S_{i-1})}{f_i^\star(x)p_i^\star(x)} \\
        &\leq \dfrac{\abs{S_{i-1}}}{\abs{S'}} \ev{x \sim \calU(S_{i-1})}{f_i^\star(x)\abs{p_i^\star(x)}} \\
        &< R\ev{x \sim \calU(S)}{f_i^\star(x) \abs{p_i^\star(x)}} + \Delta B \\
        &\leq R \cdot \dfrac{2\epsilon}{2R+\epsilon} + B \cdot \dfrac{\epsilon^2}{B(2R+\epsilon)} \tag{since \(p_i^\star \in \calP(f_i^\star)\) } \\
        &= \epsilon
    \end{align*}
    Hence, we reached a contradiction since we showed that \(\mu_i^\star \leq \epsilon\). This concludes the proof.
\end{proof}

We now prove the following proposition, which bounds the number of iterations of \Cref{alg:icf_} independently of the number of input examples. This is crucial in order to show that the selector returned by \Cref{alg:icf_} belongs to a class of bounded VC dimension, so that uniform convergence arguments can be applied to show generalization.

\begin{proposition} \label{prop:t}
    In the setting of \Cref{alg:icf_}, it holds \(t \leq \lfloor 1 / \Delta \rfloor\).
\end{proposition}

\begin{proof}
    If \(t \neq 0\), then for any \(i \in [t]\), it holds that
    \[\dfrac{\abs{S_{i-1}}}{\abs{S'}} \pr{x \sim \calU(S_{i-1})}{f_i^\star(x)\abs{p_i^\star(x)} > \tau_i^\star} \geq R\pr{x \sim \calU(S)}{f_i^\star(x)\abs{p_i^\star(x)} > \tau_i^\star} + \Delta \, , \]
    which implies that \(\abs{\{x \in S_{i-1} : f_i^\star(x)\abs{p_i^\star(x)} > \tau_i^\star \}} \geq \Delta \abs{S'}\), that is, in each iteration, at least \(\Delta \abs{S'}\) elements of \(S_0\) (that had not been discarded in any previous iteration) are discarded. Thus, we have \(\Delta \abs{S'}t \leq \abs{S_0} \leq \abs{S'}\), which implies that \(t \leq \lfloor 1 / \Delta \rfloor\).
\end{proof}

\begin{theorem}[Iterative Chow Filtering] \label{thm:icf_}
        Let any \(\ell, k \in \natural^\ast\), \(A \geq 1\), \(R > 1\), \(\beta > 0\), \(\epsilon, \delta \in (0,1)\) and let \(\calF\) be a class of \(k\) functions \(f_1, \dots, f_k \colon \real^d \to \{0, 1\}\). Let \(\calD\) be any \(A\)-hypercontractive distribution on \(\real^d\) and let \(\calD'\) be any distribution on \(\real^d\). Let \(S\) and \(S'\) be multisets of independent examples drawn from \(\calD\) and \(\calD'\) respectively. If \(\abs{S}, \abs{S'} \geq \mathrm{poly}(R, \beta, A^\ell, (d+1)^\ell, 1/\epsilon, \log(k/\delta)^\ell ) \), then \Cref{alg:icf_}, given \(\calF, S, S', \ell, R, \beta, \epsilon\) as input, runs in time \(k \cdot \mathrm{poly}(\abs{S}, \abs{S'}, (d+1)^\ell)\) and outputs a succinct \(k \sqrt{\beta R^3/\epsilon^5} \mathrm{poly} ( \abs{S}, (d+1)^\ell)\)-time-computable description of a function \(s \colon \real^d \to \{0, 1\}\) that satisfies the following guarantees.
    \begin{enumerate}
        \item \label{item:selector1_} Let any \(p_1, \dots, p_k \in \calP_{d, \ell}\) such that \(\max_{i \in [k]}\ev{x \sim \calD}{p_i(x)^2} \leq \beta\). Then, with probability at least \(1 - \delta\) over \(S\) and \(S'\), we have that for any \(i \in [k]\), it holds
        \[
        \ev{x \sim \calD'}{f_i(x)p_i(x)s(x)} \leq (R+\epsilon)(1+\epsilon)\ev{x \sim \calD}{f_i(x)\abs{p_i(x)}} + 2\epsilon \, .
        \] 
        \item \label{item:selector2_} With probability at least \(1 - \delta\) over \(S\) and \(S'\), it holds
        \[
        \pr{x \sim \calD}{s(x)=0} \leq \min \left\{ \dfrac{1+\epsilon}{R}, \dfrac{d_\mathrm{TV}( \calD, \calD')+\epsilon}{R-1}  \right\} \, .
        \]
    \end{enumerate}
\end{theorem}

\begin{proof}
Let
\[s, B, \Delta, t, S_0, \dots, S_t, \mu_1^\star, \dots, \mu_{t+1}^\star, f_1^\star, \dots, f_{t+1}^\star, p_1^\star, \dots, p_{t+1}^\star, \tau_1^\star, \dots, \tau_{t}^\star, \calP(f_1), \dots, \calP(f_k)\]
be defined as in the description of \Cref{alg:icf_}.

First, we show that \Cref{alg:icf_} is efficient. For any \(f \in \calF\), the set \(\calP(f)\) can be described by \(O(\abs{S})\) linear constraints plus one convex quadratic constraint over the coefficient vectors whose dimension is at most \((d+1)^\ell\). Thus, \cref{line:icf_:S0} and \cref{line:icf_:obj} can be implemented via convex programming in time \(k \cdot \mathrm{poly}(\abs{S}, \abs{S'}, (d+1)^\ell)\).\footnote{In \cref{line:icf_:S0,line:icf_:obj}, the algorithm solves one convex optimization problem for each \(f \in \calF\).} Moreover, \cref{line:icf_:tau} and \cref{line:icf_:Si} can be implemented with a single pass of the input points. Finally, for any \(x \in \real^d\), the value of the first indicator function in the expression of \(s(x)\) can be computed via convex programming in time \(k \cdot \mathrm{poly}\left(\abs{S}, (d+1)^\ell\right)\) and the value of each of the other \(t\) indicator functions can be computed in time \((d+1)^\ell\). So, \(s(x)\) can be computed in time \(k \sqrt{\beta R^3/\epsilon^5} \mathrm{poly} \left( \abs{S}, (d+1)^\ell\right)\).

\paragraph{Proof of \Cref{item:selector1_}.}

We proceed by proving \Cref{item:selector1_} of \Cref{thm:icf_}. Let any \(f \in \calF\) and \(p \in \calP_{d, \ell}\) such that \(\ev{x \sim \calD}{p(x)^2} \leq \beta\). Let \(\gamma = \epsilon/(R +\epsilon)\).
    
    First, we assume that \(\ev{x \sim \calD}{f(x)\abs{p(x)}} \leq \gamma\). Then, if \(\abs{S} \gtrsim R^2\beta(d+1)^{2\ell} \left(CA\log(k/\delta) \right)^{4\ell+1}/\epsilon^2\) for some sufficiently large constant \(C > 0\), then, by \Cref{lemma:con2}, we have that with probability at least \(1 - \delta/(2k)\) over \(S\), it holds \(p \in \calP(f)\). We condition on that happening for the rest of the paragraph. Then, by definition of the selector \(s\), it holds \(s(x) = 1\) for all \(x \in S_t\) and \(s(x) = 0\) for all \(x \in S' \setminus S_t\). Thus, we have
    \begin{align*}
    \ev{x \sim \calU(S')}{f(x)p(x)s(x)} &= \dfrac{1}{\abs{S'}} \sum_{x \in S'} f(x)p(x)s(x) \\
    &= \dfrac{1}{\abs{S'}} \sum_{x \in S_{t}} f(x)p(x) \\
    &\leq \dfrac{1}{\abs{S'}} \sum_{x \in S_{t}} f_{t+1}^\star(x)p_{t+1}^\star(x) \\
    &\leq \epsilon \, .
    \end{align*}
    The penultimate inequality follows from the fact that \(p \in \calP(f)\) and from the fact that the pair \((f_{t+1}^\star, p_{t+1}^\star)\) is a maximizer of the optimization problem defined in \cref{line:icf_:obj} of \Cref{alg:icf_} and the final inequality follows from the fact that \(t+1\) is the last iteration of \Cref{alg:icf_} (thus, the terminating condition \(\mu_{t+1}^\star \leq \epsilon\) of \cref{line:icf_:terminate} of \Cref{alg:icf_} holds). 
    Now, consider the class
    \[
    \calH = \left\{ x \mapsto \INDICATOR{\max_{g \in \calF}\max_{q \in \calP(g)} g(x)\abs{q(x)} \leq B} \prod_{i=1}^{\lfloor 1 / \Delta \rfloor}\indicator{f_i(x) \abs{p_i(x)} \leq \tau_i} :
    \begin{array}{l}
       f_i \in \calF \\
       p_i \in \calP_{d, \ell} \\
       \tau_i \geq 0
    \end{array} \right\} \, .
    \]
    Notice that \(\calH\) does not depend on \(S'\). Also, due to \Cref{prop:t}, we have that \(t \leq \lfloor 1/\Delta \rfloor\), which implies that \(s \in \calH\). Moreover, each function in \(\calH\) is a logical AND of \(\lfloor 1 / \Delta \rfloor + 1\) functions, the first of which is a fixed function (after conditioning on \(S\)) and the other \(\lfloor 1 / \Delta \rfloor\), by \Cref{lemma:vc1}, belong to a class of VC dimension at most \( O\left( (d+1)^\ell \ell \log(d)\log(k) \right)\). Thus, by \Cref{fact:vc_and}, we have that \(\mathrm{VCdim}(\calH) \leq O\left( (d+1)^\ell \ell \log(d)\log(k) \right) \lfloor 1 / \Delta \rfloor \log(\lfloor 1 / \Delta \rfloor) \). Moreover, since \(\sup_{h \in \calH}\sup_{x \in \real^d} h(x) f(x)\abs{p(x)} \leq B\),
    if \(\abs{S'} \gtrsim B^4  \mathrm{VCdim}(\calH) \log(1/\epsilon) \log(k/\delta)/\epsilon^8 \), \Cref{lemma:uc1} (uniform convergence) implies that with probability at least \(1 - \delta/(2k)\) over \(S'\), it holds
    \[
    \ev{x \sim \calD'}{f(x)p(x)s(x)} \leq \ev{x \sim \calU(S')}{f(x)p(x)s(x)} + \epsilon^2 \leq \epsilon + \epsilon^2 \, . 
    \]
    Thus, \(f\) and \(p\) satisfy the condition of \Cref{item:selector1_} with probability at least \(1 - \delta/k\).

    Now, we assume that \(\ev{x \sim \calD}{f(x)\abs{p(x)}} > \gamma\). We define the polynomial
    \(
    \widetilde p = \gamma p/\ev{x \sim \calD}{f(x)\abs{p(x)}}
    \). Then, \(\ev{x \sim \calD}{\widetilde p(x)^2} < \beta\) and \(\ev{x \sim \calD}{f(x)\abs{\widetilde p(x)}} = \gamma\). If \(\abs{S}\) is sufficiently large, by \Cref{lemma:con2}, we have that with probability at least \(1 - \delta/(2k)\) over \(S\), it holds \(\widetilde p \in \calP(f)\). We condition on this high probability event for the rest of the argument. Following the same steps as in the previous case, we have that with probability at least \(1 - \delta/(2k)\) over \(S'\), it holds
    \(
    \ev{x \sim \calD'}{f(x)\widetilde p(x)s(x)} \leq \epsilon + \epsilon^2\), which implies that
    \[
    \ev{x \sim \calD'}{f(x)p(x)s(x)} \leq (R + \epsilon)(1+\epsilon)\ev{x \sim \calD}{f(x)\abs{p(x)}} \, . 
    \]
    Thus, \(f\) and \(p\) again satisfy the condition of \Cref{item:selector1_} with probability at least \(1 - \delta/k\). Finally, a union bound over all pairs \((f_i, p_i)\) for \(i \in [k]\), implies that with probability at least \(1 - \delta\), it holds
    \[
    \ev{x \sim \calD'}{f_i(x)p_i(x)s(x)} \leq (R+\epsilon)(1+\epsilon)\ev{x \sim \calD}{f_i(x)\abs{p_i(x)}} + 2\epsilon
    \]
    for all \(i \in [k]\).

    \paragraph{Proof of \Cref{item:selector2_}.} Now, we will prove \Cref{item:selector2_} of \Cref{thm:icf_}. In the trivial case \(t = 0\), we have \(s(x) = \indicator{\max_{f \in \calF}\max_{p \in \calP(f)} f(x)\abs{p(x)} \leq B}\). So, we have
    \[
        \pr{x \sim \calD}{s(x)=0} = \PR{x \sim \calD}{\max_{f \in \calF}\max_{p \in \calP(f)} f(x)\abs{p(x)} > B} \, .
    \]
    We now bound the right hand side in the above equality. Let \(\calP = \{p \in \calP_{d, \ell} : \ev{x \sim \calU(S)}{p(x)^2 } \leq 2\beta\}\) and recall that \(B = 2\sqrt{2R(d+1)^\ell \beta/\epsilon} \). Hence,
    if \(\abs{S} \gtrsim (d+1)^{2\ell} (4A\log(12/\delta) )^{4\ell+1}\), by \Cref{lemma:con1}, we have that with probability at least \(1 - \delta/3\) over \(S\), it holds
    \[
    \PR{x \sim \calD}{\max_{f \in \calF}\max_{p \in \calP(f)} f(x)\abs{p(x)} > B} \leq \PR{x \sim \calD}{\max_{p \in \calP} \abs{p(x)} > B} \leq \epsilon/(2R) \, ,
    \]
    which directly implies \Cref{item:selector2_} of \Cref{thm:icf_}.
    
    We now consider the case \(t \geq 1\). By definition of the selector, we have
    \[
        \pr{x \sim \calD}{s(x)=0} \leq \PR{x \sim \calD}{\max_{f \in \calF}\max_{p \in \calP(f)} f(x)\abs{p(x)} > B} + \sum_{i=1}^{t} \pr{x \sim \calD}{f_i^\star(x) \abs{p_i^\star(x)} > \tau_i^\star} \, .
    \]
    We now bound the right term of the right hand side in the above inequality. If \(\abs{S} \gtrsim (d+1)^\ell \ell \log(d) \log(k) \allowbreak \log(R/\Delta) \log(1/\delta)R^2/\Delta^2\), by \Cref{lemma:vc1} and \Cref{fact:uc0} (uniform convergence), we have that with probability at least \(1 - \delta/3\) over \(S\), for all \(i \in [t]\), it holds
    \[
    \pr{x \sim \calD}{f_i^\star(x) \abs{p_i^\star(x)} > \tau_i^\star} \leq  \pr{x \sim \calU(S)}{f_i^\star(x)\abs{p_i^\star(x)} > \tau_i^\star} + \dfrac{\Delta}{R} \, .
    \]
    Hence, summing over all \(i \in [t]\), we get that
    \begin{align*}
    \sum_{i=1}^{t} \pr{x \sim \calD}{f_i^\star(x) \abs{p_i^\star(x)} > \tau_i^\star} &\leq \sum_{i=1}^{t} \left( \pr{x \sim \calU(S)}{f_i^\star(x)\abs{p_i^\star(x)} > \tau_i^\star} + \dfrac{\Delta}{R} \right) \\
    &\leq \dfrac{1}{R}\sum_{i=1}^{t} \dfrac{\abs{S_{i-1}}}{\abs{S'}} \pr{x \sim \calU(S_{i-1})}{f_i^\star(x)\abs{p_i^\star(x)} > \tau_i^\star} \\
    &= \dfrac{1}{R\abs{S'}} \sum_{i=1}^{t} \sum_{x \in S_{i-1}} \indicator{f_i^\star(x)\abs{p_i^\star(x)} > \tau_i^\star} \\
    &= \dfrac{1}{R\abs{S'}} \sum_{i=1}^t \sum_{x \in S_{i-1} \setminus S_i} \indicator{f_i^\star(x)\abs{p_i^\star(x)} > \tau_i^\star} \, .
    \end{align*}
    In the second line, we used the fact that the condition of \cref{line:icf_:tau} of \Cref{alg:icf_} is satisfied. In the fourth line, we used the fact that for all \(i \in [t]\), we have \(\sum_{x \in S_i} \indicator{f_i^\star(x)\abs{p_i^\star(x)} > \tau_i^\star} = 0\) by definition of \(S_i\) (see \cref{line:icf_:Si} of \Cref{alg:icf_}). Using the fact that \(\indicator{f_i^\star(x)\abs{p_i^\star(x)} > \tau_i^\star} \leq \indicator{s(x)=0}\) for all \(i \in [t]\) and \(x \in \real^d\), we get
    {\allowdisplaybreaks
    \begin{align*}
    \sum_{i=1}^{t} \pr{x \sim \calD}{f_i^\star(x) \abs{p_i^\star(x)} > \tau_i^\star} &\leq
    \dfrac{1}{R\abs{S'}} \sum_{i=1}^t \sum_{x \in S_{i-1} \setminus S_i} \indicator{s(x)=0} \\
    &= \dfrac{1}{R\abs{S'}} \sum_{x \in S_0} \indicator{s(x)=0} \\
    &\leq \dfrac{1}{R\abs{S'}} \sum_{x \in S'} \indicator{s(x)=0} \\
    &= \dfrac{1}{R} \pr{x \sim \calU(S')}{s(x)=0} \, .
    \end{align*}
    }
    Trivially, the quantity above is upper bounded by \(1/R\), which implies that \(\pr{x \sim \calD}{s(x)=0} \leq (1 + \epsilon)/R\). We will now provide an additional upper bound for the same quantity. If \(\abs{S'} \gtrsim  \mathrm{VCdim}(\calH) \log(1/\epsilon) \log(1/\delta)/\epsilon^2 \), using \Cref{fact:uc0} (uniform convergence), we get that with probability at least \(1 - \delta/3\), it holds
    \[
        \pr{x \sim \calU(S')}{s(x)=0} \leq \pr{x \sim \calD'}{s(x)=0} + \epsilon/2
        \leq d_\mathrm{TV}( \calD, \calD' ) + \pr{x \sim \calD}{s(x)=0} + \epsilon/2 \, .
    \]
    Thus, we have
    \[
    \pr{x \sim \calD}{s(x)=0} \leq \dfrac{1}{R} \left( d_\mathrm{TV}( \calD, \calD') + \pr{x \sim \calD}{s(x)=0} + \epsilon \right) \, ,
    \]
    which implies (after rearranging) that
    \[
    \pr{x \sim \calD}{s(x)=0} \leq \dfrac{d_\mathrm{TV}(\calD, \calD') + \epsilon}{R-1} \, .
    \]
    Finally, a union bound implies that the condition of \Cref{item:selector2_} holds with probability at least \(1 - \delta\). This concludes the proof of \Cref{thm:icf_}.
\end{proof}

\section{PQ Learning} \label{appendix:pq}

In this section, we prove \Cref{thm:pq}, which we restate here as \Cref{thm:pq_} for convenience.

\begin{theorem}[PQ Learning via \(\calL_1\) Sandwiching] \label{thm:pq_} Let any \(\epsilon, \delta, \eta \in (0, 1)\) and \(A \geq 1\). Let \(\calX \subseteq \real^d \) and let \(\calD\) be any \(A\)-hypercontractive distribution on \(\calX\). Let \(\calC \subseteq \{0, 1\}^\calX\) be a concept class and let \(\ell=\calL_1(\calC, \calD, \epsilon\eta/8)\). Then, there exists an \((\epsilon, \delta, \eta)\)-PQ learner for \(\calC\) with respect to \(\calD\) with sample complexity and runtime \(\mathrm{poly} ((d+1)^\ell, A^\ell, 1/\epsilon, 1/\eta, \log(1/\delta)^\ell )\).
\end{theorem}

By combining the above result with existing sandwiching degree bounds from prior work (see \Cref{appendix:sandwiching-bounds}), we obtain the following corollaries.

\begin{corollary} Let any \(\epsilon, \delta, \eta \in (0, 1)\). Let \(\calC\) be the class of depth-\(t\), size-\(s\) \(\mathsf{AC}^0\) circuits in \(\{\pm 1\}^d\) and let \(\ell = O(\log(s))^{O(t)}\log(1/(\epsilon\eta))\). Then, there exists an \((\epsilon, \delta, \eta)\)-PQ learner for \(\calC\) with respect to \(\calU(\{\pm 1\}^d)\) with sample complexity and runtime \(\mathrm{poly} ((d+1)^\ell, 1/\epsilon, 1/\eta, \log(1/\delta)^\ell )\).
\end{corollary}

\begin{corollary} Let any \(\epsilon, \delta, \eta \in (0, 1)\). Let \(\calC\) be the class of depth-\(t\), size-\(s\) decision trees of halfspaces in \(\real^d\) and let \(\ell = \tilde O(t^4s^2/(\epsilon\eta)^2)\). Then, there exists an \((\epsilon, \delta, \eta)\)-PQ learner for \(\calC\) with respect to \(\calN_d\) (resp. \(\calU(\{\pm 1\}^d)\)) with sample complexity and runtime \(\mathrm{poly} ((d+1)^\ell, 1/\epsilon, 1/\eta, \log(1/\delta)^\ell )\).
\end{corollary}

\begin{corollary} Let any \(\epsilon, \delta, \eta \in (0, 1)\). Let \(\calC\) be the class of degree-\(2\) PTFs in \(\real^d\) and let \(\ell = \tilde O(1/(\epsilon\eta)^8)\) (resp. \(\ell = \tilde O(1/(\epsilon\eta)^9)\)). Then, there exists an \((\epsilon, \delta, \eta)\)-PQ learner for \(\calC\) with respect to \(\calN_d\) (resp. \(\calU(\{\pm1\}^d)\)) with sample complexity and runtime \(\mathrm{poly} ((d+1)^\ell, 1/\epsilon, 1/\eta, \log(1/\delta)^\ell )\).
\end{corollary}

\begin{corollary} Let any \(\epsilon, \delta, \eta \in (0, 1)\). Let \(\calC\) be the class of degree-\(k\) PTFs in \(\real^d\) and let \(\ell = O_k((\epsilon\eta)^{-4k \cdot 7^k})\). Then, there exists an \((\epsilon, \delta, \eta)\)-PQ learner for \(\calC\) with respect to \(\calN_d\) with sample complexity and runtime \(\mathrm{poly} ((d+1)^\ell, 1/\epsilon, 1/\eta, \log(1/\delta)^\ell )\).
\end{corollary}

\begin{corollary} Let any \(\epsilon, \delta, \eta \in (0, 1)\). Let \(\calC\) be the class of arbitrary functions of \(k\) halfspaces in \(\real^d\), let \(\calD\) be any isotropic log-concave distribution on \(\real^d\) and let \(\ell = \exp((\log(\log(k)/(\epsilon\eta)))^{O(k)}/(\epsilon\eta)^4)\). Then, there exists an \((\epsilon, \delta, \eta)\)-PQ learner for \(\calC\) with respect to \(\calD\) with sample complexity and runtime \(\mathrm{poly} ((d+1)^\ell, 1/\epsilon, 1/\eta, \log(1/\delta)^\ell )\).
\end{corollary}

We will now prove \Cref{thm:pq_}.

\begin{proof}[Proof of \Cref{thm:pq_}]
    Let \(\calD^\mathrm{train}\) be any training distribution on \(\calX \times \{0, 1\}\) whose marginal on \(\calX\) is \(\calD\). Let \(\calD^\mathrm{test}\) be any test distribution on \(\calX \times \{0, 1\}\) and let \(\calD'\) denote its marginal on \(\calX\). Let \(\mathsf{opt}_\mathrm{train} = \inf_{f \in \calC}\pr{(x,y) \sim \calD^\mathrm{train}}{f(x) \neq y}\). The algorithm does the following:
    \begin{enumerate}
        \item It draws a set \(S_\mathrm{train}\) of \(\mathrm{poly} \left((d+1)^\ell, 1/\epsilon, 1/\eta, \log(1/\delta) \right)\) i.i.d. examples from \(\calD^\mathrm{train}\).
        \item It runs the degree-\(\ell\) \(\calL_1\) polynomial regression algorithm of \Cref{thm:l1_regression} with input \(S_\mathrm{train}\) to get a classifier \(\hat f \colon \real^d \to \{0, 1\}\).
        \item It draws sets \(S, S'\) of \(\mathrm{poly} ((d+1)^\ell, A^\ell, 1/\epsilon, 1/\eta, \log(1/\delta)^\ell )\) i.i.d. examples from \(\calD\) and \(\calD'\) respectively.
        \item It runs \Cref{alg:icf} with input \(\hat f, S, S', \ell, R'=1/\eta+\epsilon/96, \beta= 4(2A)^{2\ell}, \epsilon'=\epsilon\eta/96\) to get a selector \(s \colon \calX \to \{0, 1\}\).
        \item It returns \((\hat f, s)\).
    \end{enumerate}

    First, since \(\abs{S_\mathrm{train}}\) is sufficiently large, by \Cref{thm:l1_regression}, we can ensure that with probability at least \(1 - \delta/2\), it holds \(\pr{(x, y) \sim \calD^\mathrm{train}}{\hat f(x) \neq y} \leq \mathsf{opt}_\mathrm{train} + \epsilon\eta/2\).

    Now, let any \(f^\star \in \calC\) such that \(\errtrain(f^\star) = \lambda_\mathrm{train}\) and \(\errtest(f^\star) = \lambda_\mathrm{test}\) and let \(p_\mathrm{up}, p_\mathrm{down}\) be \(\epsilon''\)-approximate \(\calL_1\) sandwiching polynomials for \(f^\star\) under \(\calD\) of degree at most \(\ell\), where \(\epsilon'' = \epsilon\eta/8\). Then, \(p_\mathrm{up}\) and \(\bar p_\mathrm{down}\) are nonnegative. Moreover, using the fact that \(\calD\) is \(A\)-hypercontractive and the fact that \(\ev{x \sim \calD}{p_\mathrm{up}(x) - f^\star(x)} \leq \epsilon'' \leq 1\), we get that
    \begin{align*}
        \ev{x \sim \calD}{p_\mathrm{up}(x)^2} &\leq (2A)^{2\ell}\ev{x \sim \calD}{p_\mathrm{up}(x)}^2 \\ &\leq (2A)^{2\ell} \left(\ev{x \sim \calD}{p_\mathrm{up}(x) - f^\star(x)} + \ev{x \sim \calD}{f^\star(x)} \right)^2 \\
        &< 4(2A)^{2\ell}
    \end{align*}
    and, similarly, \(\ev{x \sim \calD}{\bar p_\mathrm{down}(x)^2} < 4(2A)^{2\ell}\). Since \(\abs{S}\) and \(\abs{S'}\) are sufficiently large, applying \Cref{thm:icf} (with additive error \(\epsilon'\), multiplicative slack \(R'\) and polynomial variance bound \(\beta = 4(2A)^{2\ell}\)), we get that with probability at least \(1 - \delta/2\) we have that \(\pr{x \sim \calD}{s(x)=0} \leq (1+\epsilon')/R'\) and
    \begin{align*}
        \ev{x \sim \calD'}{\hat f(x)\bar p_\mathrm{down}(x)s(x)} &\leq (R'+\epsilon')(1+\epsilon')\ev{x \sim \calD}{\hat f(x)\bar p_\mathrm{down}(x)} + 2\epsilon' \\
        &\leq R'\ev{x \sim \calD}{\hat f(x)\bar p_\mathrm{down}(x)} + 8R'\epsilon'
    \end{align*}
    and
    \begin{align*}
        \ev{x \sim \calD'}{\bar{\hat f}(x)\pup(x)s(x)} &\leq (R'+\epsilon')(1+\epsilon')\ev{x \sim \calD}{ \bar{\hat f}(x)\pup(x)} + 2\epsilon' \\
        &\leq R'\ev{x \sim \calD}{ \bar{\hat f}(x)\pup(x)} + 8R'\epsilon' \, ,
    \end{align*}
    where we used the fact that \(\ev{x \sim \calD}{\hat f(x)\bar p_\mathrm{down}(x)} \leq \ev{x \sim \calD}{\bar p_\mathrm{down}(x)} = \ev{x \sim \calD}{f^\star(x) - p_\mathrm{down}(x)} + \ev{x \sim \calD}{\bar f^\star(x)} \leq 2\) and, similarly, \(\ev{x \sim \calD}{\bar{\hat f}(x) p_\mathrm{up}(x)} \leq 2\) in the last inequality of each of the above lines. Thus, for \(R' = 1/\eta + \epsilon/96\) and \(\epsilon' = \epsilon\eta/96\), we have that \(\pr{x \sim \calD}{s(x)=0} \leq \eta\) and
    \begin{align*}
        \ev{x \sim \calD'}{\hat f(x)\bar p_\mathrm{down}(x)s(x)} &\leq \eta^{-1} \ev{x \sim \calD}{\hat f(x)\bar p_\mathrm{down}(x)} + \epsilon/8 \\
        \ev{x \sim \calD'}{\bar{\hat f}(x)\pup(x)s(x)} &\leq \eta^{-1} \ev{x \sim \calD}{ \bar{\hat f}(x)\pup(x)} + \epsilon/8 \, .
    \end{align*}
    Hence, by \Cref{lemma:chow_transfer} (with \(R = 1/\eta\)), given that \(\pr{(x, y) \sim \calD^\mathrm{train}}{\hat f(x) \neq y} \leq \mathsf{opt}_\mathrm{train} + \epsilon\eta/2\), we have
    \begin{align*}
        \pr{(x, y) \sim \calD^\mathrm{test}}{ \hat f(x) \neq y \land s(x)=1 } &\leq \lambda_\mathrm{test} + \eta^{-1} \left( \lambda_\mathrm{train} + \pr{(x, y) \sim \calD^\mathrm{train}}{ \hat f(x) \neq y } \right) + \epsilon/2 \\
        &\leq \lambda_\mathrm{test} + \eta^{-1} \left( \lambda_\mathrm{train} + \mathsf{opt}_\mathrm{train} \right) + \epsilon \, .
    \end{align*}
    
    A union bound implies that, with probability at least \(1 - \delta\), we have
    \begin{enumerate}
        \item \(\pr{(x,y) \sim \calD^\mathrm{test}}{ \hat f (x) \neq y \land s(x)=1 } \leq \lambda_\mathrm{test} + (\lambda_\mathrm{train} + \mathsf{opt}_\mathrm{train})/\eta + \epsilon\) and
        \item \(\pr{x \sim \calD}{s(x) = 0} \leq \eta\),
    \end{enumerate}
    which concludes the proof.
\end{proof}

\section{Tolerant TDS Learning} \label{appendix:tolerant_tds}

In this section, we provide our result for tolerant TDS learning. As before, we consider a feature space $\calX\subseteq \real^d$, binary labels $\calY=\{0,1\}$ and the learner is assumed to have access to labeled examples from a training distribution $\calD^{\mathrm{train}}$ over $\calX\times\calY$, as well as unlabeled examples from the marginal $\calD^{\mathrm{test}}_\calX$ of a test distribution $\calD^{\mathrm{test}}$. The formal definition of tolerant TDS learning of a concept class $\calC\subseteq\{\calX\to \calY\}$ with respect to a distribution $\calD$ over $\calX$ is as follows.
\begin{definition}
[Tolerant TDS Learning \cite{goel2024tolerantalgorithmslearningarbitrary, klivans2024testablelearningdistributionshift}]
\label{def:tolerant_tds}
Let $\delta\in(0,1)$ and \(\theta \in [0, 1]\). An algorithm $\calA$ is a TDS learner for $\calC$ with respect to $\calD$ up to error \(\gamma\), tolerance \(\theta\) and failure probability \(\delta\) if the following holds whenever $\Dtrain_\calX=\calD$. Given sufficiently many labeled examples from $\Dtrain$ and sufficiently many unlabeled examples from $\Dtest_\calX$, the algorithm either accepts and outputs a classifier \(h \colon \calX \to \calY\) or rejects. Moreover:
\begin{enumerate}
    \item \emph{Soundness.}
    If \(\calA\) accepts, then \(
    \pr{(x,y) \sim \calD^\mathrm{test}}{h(x) \neq y} \leq \gamma
    \) with probability at least \(1 - \delta\).
    \item \emph{Completeness.}
    If \(d_\mathrm{TV}(\calD_\calX^\mathrm{train}, \calD_\calX^\mathrm{test}) \leq \theta\), then \(\calA\) accepts with probability at least \(1 - \delta\).
\end{enumerate}
The error benchmark $\gamma$ may depend on the particular instance of the problem, i.e., the distributions $\Dtrain,\Dtest$, but its value is not necessarily known to the learner.
\end{definition}

We first note that our result for PQ learning already implies that \(\calL_1\) sandwiching suffices for tolerant TDS learning due to the following result of \cite{goel2024tolerantalgorithmslearningarbitrary}.

\begin{proposition}[PQ Learning Implies Tolerant TDS Learning \cite{goel2024tolerantalgorithmslearningarbitrary}] Let \(\gamma, \delta \in (0, 1)\). Let \(\calA\) be an algorithm that PQ-learns \(\calC\) with respect to  \(\calD\) up to error \(\gamma\), rejection rate \(\eta\) and failure probability \(\delta\).\footnote{Here, we mean that \(\calA\) outputs some \(h, s \colon \calX \to \{0, 1\}\) such that, with probability at least \(1 - \delta\), it holds \(\pr{(x, y) \sim \Dtest}{h(x) \neq y \land s(x)=1} \leq \gamma\) and \(\pr{x \sim \calD_\calX^\mathrm{train}}{s(x)=0} \leq \eta\).} Then, for any \(\epsilon \in (0, 1)\) and \(\theta \in [0, 1)\), there exists an algorithm that TDS-learns \(\calC\) with respect to \(\calD\) up to error \(\gamma+\eta+\theta+\epsilon\), tolerance \(\theta\) and failure probability \(2\delta\) that
calls \(\calA\) once and uses \(O(\log(1/\delta)/\epsilon^2)\) additional examples and evaluations of the selector given by \(\calA\).
\end{proposition}

\begin{corollary} \label{cor:tds_naive}
Let any \(\epsilon, \delta, \eta \in (0, 1)\) and \(\theta \in [0, 1)\). Let \(\calX \subseteq \real^d \) and let \(\calD\) be any \(A\)-hypercontractive distribution on \(\calX\). Let \(\calC \subseteq \{0, 1\}^\calX\) be a concept class and let \(\ell=\calL_1(\calC, \calD, \epsilon\eta/9)\). Then, there exists a TDS learner for \(\calC\) with respect to \(\calD\) up to error \(\lambda_\mathrm{test} + (\lambda_\mathrm{train} + \mathsf{opt}_\mathrm{train})/\eta +\eta + \theta + \epsilon\), tolerance \(\theta\) and failure probability \(\delta\) with sample complexity and runtime \(\mathrm{poly} ((d+1)^\ell, A^\ell, 1/\epsilon, 1/\eta, \log(1/\delta)^\ell )\).
\end{corollary}

Notice that setting \(\eta = \theta + \epsilon\) in \Cref{cor:tds_naive}, yields a TDS learner with error \(\lambda_\mathrm{test} + \frac{\lambda_\mathrm{train} + \mathsf{opt}_\mathrm{train}}{\theta +\epsilon} + 2\theta + 2\epsilon\) and tolerance \(\theta\). To obtain an error bound with improved dependence in \(\lambda_\mathrm{train} + \mathsf{opt}_\mathrm{train}\), we provide a TDS learner that does not use a PQ learner as a black box, but directly invokes the ICF algorithm.

Our TDS learner (see \Cref{thm:tds} for a detailed description) first learns a classifier \(\hat f\) via \(\calL_1\) polynomial regression on a sample from \(\Dtrain\), then runs ICF using \(\hat f\) and unlabeled data from \(\Dtrain_\calX\) and \(\Dtest_\calX\) as input to obtain a selector and, finally accepts or rejects depending on the selector's empirical rejection rate on unlabeled data from \(\Dtest_\calX\). A key technical ingredient in the analysis of our algorithm is the improved bound of \Cref{thm:icf_} on the rejection rate of the selector, which upper bounds the rejection rate by a multiple of the total variation distance between the training and test marginals.

Finally, we remark that our tolerant TDS learner is similar in spirit to the tolerant TDS
learner of \cite{goel2024tolerantalgorithmslearningarbitrary}. The difference is that \cite{goel2024tolerantalgorithmslearningarbitrary} use \(\calL_2\) polynomial regression to learn the classifier \(\hat f\). Second, instead of ICF, their algorithm uses their spectral outlier removal method (which enforces their \(\calL_2\) sandwiching assumptions) to obtain a selector.

\begin{theorem}[Tolerant TDS Learning via \(\calL_1\) Sandwiching] \label{thm:tds} Let any \(\epsilon, \delta \in (0, 1)\), \(\theta \in [0, 1)\), \(R > 1\) and \(A \geq 1\). Let \(\calX \subseteq \real^d \) and let \(\calD\) be any \(A\)-hypercontractive distribution on \(\calX\). Let \(\calC \subseteq \{0, 1\}^\calX\) be a concept class and let \(\ell = \calL_1(\calC, \calD, \epsilon/(16R))\). Then, there is an algorithm that TDS-learns \(\calC\) with respect to \(\calD\) up to error \(\lambda_\mathrm{test} + R(\lambda_\mathrm{train} + \mathsf{opt}_\mathrm{train}) + R(R-1)^{-1}\theta + \epsilon\), tolerance \(\theta\) and failure probability \(\delta\) in time \(\mathrm{poly} ((d+1)^\ell, A^\ell, R(R-1)^{-1}/\epsilon, \log(1/\delta)^\ell )\).
\end{theorem}

\begin{remark} \label{remark:tds}
    Setting \(R = 1 + \max\{ \sqrt{\theta/2}, \epsilon/9\}\) above, we get that for any class \(\calC\), there exists an algorithm that TDS-learns \(\calC\) with respect to \(\calD\) up to error \(\lambda + \mathsf{opt}_\mathrm{train} + 4\sqrt{\theta} + \epsilon\), tolerance \(\theta\) and failure probability \(\delta\) in time \(\mathrm{poly} ((d+1)^\ell, A^\ell, 1/\epsilon, \log(1/\delta)^\ell )\), where \(\ell = \calL_1(\calC, \calD, \epsilon/64)\).
\end{remark}

By combining \Cref{remark:tds} with existing sandwiching degree bounds from prior work (see \Cref{appendix:sandwiching-bounds}), we obtain the following corollaries.

\begin{corollary} Let any \(\epsilon, \delta \in (0, 1)\) and \(\theta \in [0, 1)\). Let \(\calC\) be the class of depth-\(t\), size-\(s\) \(\mathsf{AC}^0\) circuits in \(\{\pm 1\}^d\) and let \(\ell = O(\log(s))^{O(t)}\log(1/(\epsilon))\). Then, there is an algorithm that TDS-learns \(\calC\) with respect to \(\calU(\{\pm1\}^d)\) up to error \(\lambda + \mathsf{opt}_\mathrm{train} + 4\sqrt{\theta} + \epsilon\), tolerance \(\theta\) and failure probability \(\delta\) in time \(\mathrm{poly} ((d+1)^\ell, 1/\epsilon, \log(1/\delta)^\ell )\).
\end{corollary}

\begin{corollary} Let any \(\epsilon, \delta \in (0, 1)\) and \(\theta \in [0, 1)\). Let \(\calC\) be the class of depth-\(t\), size-\(s\) decision trees of halfspaces in \(\real^d\) and let \(\ell = \tilde O(t^4s^2/\epsilon^2)\). Then, there is an algorithm that TDS-learns \(\calC\) with respect to \(\calN_d\) (resp. \(\calU(\{\pm 1\}^d)\)) up to error \(\lambda + \mathsf{opt}_\mathrm{train} + 4\sqrt{\theta} + \epsilon\), tolerance \(\theta\) and failure probability \(\delta\) in time \(\mathrm{poly} ((d+1)^\ell, 1/\epsilon, \log(1/\delta)^\ell )\).
\end{corollary}

\begin{corollary} Let any \(\epsilon, \delta \in (0, 1)\) and \(\theta \in [0, 1)\). Let \(\calC\) be the class of degree-\(2\) PTFs in \(\real^d\) and let \(\ell = \tilde O(1/\epsilon^8)\) (resp. \(\ell = \tilde O(1/\epsilon^9)\)). Then, there is an algorithm that TDS-learns \(\calC\) with respect to \(\calN_d\) (resp. \(\calU(\{\pm1\}^d)\)) up to error \(\lambda + \mathsf{opt}_\mathrm{train} + 4\sqrt{\theta} + \epsilon\), tolerance \(\theta\) and failure probability \(\delta\) in time \(\mathrm{poly} ((d+1)^\ell, 1/\epsilon, \log(1/\delta)^\ell )\).
\end{corollary}

\begin{corollary} Let any \(\epsilon, \delta \in (0, 1)\) and \(\theta \in [0, 1)\). Let \(\calC\) be the class of degree-\(k\) PTFs in \(\real^d\) and let \(\ell = O_k(\epsilon^{-4k \cdot 7^k})\). Then, there is an algorithm that TDS-learns \(\calC\) with respect to \(\calN_d\) up to error \(\lambda + \mathsf{opt}_\mathrm{train} + 4\sqrt{\theta} + \epsilon\), tolerance \(\theta\) and failure probability \(\delta\) in time \(\mathrm{poly} ((d+1)^\ell, 1/\epsilon, \log(1/\delta)^\ell )\).
\end{corollary}

\begin{corollary} Let any \(\epsilon, \delta \in (0, 1)\) and \(\theta \in [0, 1)\). Let \(\calC\) be the class of arbitrary functions of \(k\) halfspaces in \(\real^d\), let \(\calD\) be any isotropic log-concave distribution on \(\real^d\) and let \(\ell = \exp((\log(\log(k)/\epsilon))^{O(k)}/\epsilon^4)\). Then, there is an algorithm that TDS-learns \(\calC\) with respect to \(\calD\) up to error \(\lambda + \mathsf{opt}_\mathrm{train} + 4\sqrt{\theta} + \epsilon\), tolerance \(\theta\) and failure probability \(\delta\) in time \(\mathrm{poly} ((d+1)^\ell, 1/\epsilon, \log(1/\delta)^\ell )\).
\end{corollary}

We will now prove \Cref{thm:tds}.

\begin{proof}[Proof of \Cref{thm:tds}]
    Let \(\calD^\mathrm{train}\) be any training distribution on \(\calX \times \{0, 1\}\) whose marginal on \(\calX\) is \(\calD\). Let \(\calD^\mathrm{test}\) be any test distribution on \(\calX \times \{0, 1\}\) and let \(\calD'\) denote its marginal on \(\calX\). Let \(\mathsf{opt}_\mathrm{train} = \inf_{f \in \calC}\pr{(x,y) \sim \calD^\mathrm{train}}{f(x) \neq y}\). The algorithm does the following:
    \begin{enumerate}
        \item It draws a set \(S_\mathrm{train}\) of \(\mathrm{poly} \left((d+1)^\ell, R/\epsilon, \log(1/\delta) \right)\) i.i.d. examples from \(\calD^\mathrm{train}\) and a set \(X_\mathrm{test}\) of \(O(\log(1/\delta)/\epsilon^2)\) i.i.d. examples from \(\calD'\).
        \item It runs the degree-\(\ell\) \(\calL_1\) polynomial regression algorithm of \Cref{thm:l1_regression} with input \(S_\mathrm{train}\) to get a classifier \(\hat f \colon \real^d \to \{0, 1\}\).
        \item It draws sets \(S, S'\) of \(\mathrm{poly} ((d+1)^\ell, A^\ell, 1/\epsilon, R(R-1)^{-1}, \log(1/\delta)^\ell )\) i.i.d. examples from \(\calD\) and \(\calD'\) respectively.
        \item It runs \Cref{alg:icf} with input \(\hat f, S, S', \ell, R, \beta= 4(2A)^{2\ell}, \epsilon'=(R-1)\epsilon/(128R^2)\) to get a selector \(s \colon \calX \to \{0, 1\}\).
        \item If \(\pr{x \sim \calU(X_\mathrm{test})}{s(x)=0} \geq R(R-1)^{-1}\theta + \epsilon/4\) it returns \(\mathsf{REJECT}\), otherwise it returns \((\mathsf{ACCEPT}, \hat f)\).
    \end{enumerate}
    
    \paragraph{Proof of Soundness.} We assume that the algorithm accepts. So, from Hoeffding's inequality, since \( \abs{ X_\mathrm{test} } \gtrsim \log(1/\delta)/\epsilon^2 \), with probability at least \(1 - \delta/3\), it holds
    \[
    \pr{x \sim \calD'}{s(x) = 0} \leq
    \pr{x \sim \calU(X_\mathrm{test})}{s(x)=0} + \epsilon/4 \leq \dfrac{R\theta}{R-1} + \epsilon/2 \, .
    \]
    
    Moreover, since \(\abs{S_\mathrm{train}}\) is sufficiently large, by \Cref{thm:l1_regression}, we can ensure that with probability at least \(1 - \delta/3\), it holds \(\pr{(x, y) \sim \calD^\mathrm{train}}{\hat f(x) \neq y} \leq \mathsf{opt}_\mathrm{train} + \epsilon/(4R)\).

    Now, let any \(f^\star \in \calC\) such that \(\errtrain(f^\star) = \lambda_\mathrm{train}\) and \(\errtest(f^\star) = \lambda_\mathrm{test}\) and let \(p_\mathrm{up}, p_\mathrm{down}\) be \(\epsilon''\)-approximate \(\calL_1\) sandwiching polynomials for \(f^\star\) under \(\calD\) of degree at most \(\ell\), where \(\epsilon'' = \epsilon/(16R)\). Then, \(p_\mathrm{up}\) and \(\bar p_\mathrm{down}\) are nonnegative. Moreover, using the same steps as in the proof of \Cref{thm:pq_}, we get that
    \(\ev{x \sim \calD}{p_\mathrm{up}(x)^2} < 4(2A)^{2\ell}\)
    and \(\ev{x \sim \calD}{\bar p_\mathrm{down}(x)^2} < 4(2A)^{2\ell}\). Since \(\abs{S}\) and \(\abs{S'}\) are sufficiently large, applying \Cref{thm:icf} (with additive error \(\epsilon'\), multiplicative slack \(R\) and polynomial variance bound \(\beta = 4(2A)^{2\ell}\)), we get that with probability at least \(1 - \delta/3\) we have that
    \begin{align*}
        \ev{x \sim \calD'}{\hat f(x)\bar p_\mathrm{down}(x)s(x)} &\leq (R+\epsilon')(1+\epsilon')\ev{x \sim \calD}{\hat f(x)\bar p_\mathrm{down}(x)} + 2\epsilon' \\
        &\leq R\ev{x \sim \calD}{\hat f(x)\bar p_\mathrm{down}(x)} + 8R\epsilon'
    \end{align*}
    and
    \begin{align*}
        \ev{x \sim \calD'}{\bar{\hat f}(x)\pup(x)s(x)} &\leq (R+\epsilon')(1+\epsilon')\ev{x \sim \calD}{ \bar{\hat f}(x)\pup(x)} + 2\epsilon' \\
        &\leq R\ev{x \sim \calD}{ \bar{\hat f}(x)\pup(x)} + 8R\epsilon' \, ,
    \end{align*}
    where we used the fact that \(\ev{x \sim \calD}{\hat f(x)\bar p_\mathrm{down}(x)} \leq 2\) and \(\ev{x \sim \calD}{\bar{\hat f}(x) p_\mathrm{up}(x)} \leq 2\) in the last inequality of each of the above lines. Thus, for \(\epsilon' = (R-1)\epsilon/(128R^2)\), we have that
    \begin{align*}
        \ev{x \sim \calD'}{\hat f(x)\bar p_\mathrm{down}(x)s(x)} &\leq R \ev{x \sim \calD}{\hat f(x)\bar p_\mathrm{down}(x)} + \epsilon/16 \\
        \ev{x \sim \calD'}{\bar{\hat f}(x)\pup(x)s(x)} &\leq R \ev{x \sim \calD}{ \bar{\hat f}(x)\pup(x)} + \epsilon/16 \, .
    \end{align*}
    Hence, using \Cref{lemma:chow_transfer} (with additive error \(\epsilon/4\)), given that \(\pr{(x, y) \sim \calD^\mathrm{train}}{\hat f(x) \neq y} \leq \mathsf{opt}_\mathrm{train} + \epsilon/(4R)\), we have
    \begin{align*}
        \pr{(x, y) \sim \calD^\mathrm{test}}{ \hat f(x) \neq y \land s(x)=1 } &\leq \lambda_\mathrm{test} + R \left( \lambda_\mathrm{train} + \pr{(x, y) \sim \calD^\mathrm{train}}{ \hat f(x) \neq y } \right) + \epsilon/4 \\
        &\leq \lambda_\mathrm{test} + R \left( \lambda_\mathrm{train} + \mathsf{opt}_\mathrm{train} \right) + \epsilon/2 \, .
    \end{align*}
    
    Taking a union bound and applying the law of total probability, we conclude that with probability at least \(1 - \delta\), we have
    \begin{align*}
    \pr{(x,y) \sim \calD^\mathrm{test}}{ \hat f (x) \neq y } &\leq \pr{x \sim \calD'}{ s(x)=0 } + \pr{(x,y) \sim \calD^\mathrm{test}}{ \hat f (x) \neq y \land s(x)=1 } \\
    &\leq \lambda_\mathrm{test} + R( \lambda_\mathrm{train} + \mathsf{opt}_\mathrm{train} ) + \dfrac{R\theta}{R-1} + \epsilon \, .
    \end{align*}
    
    \paragraph{Proof of Completeness.} We assume that \(d_\mathrm{TV}(\calD, \calD') \leq \theta\). Since \(\abs{S}\) and \(\abs{S'}\) are sufficiently large, by \Cref{thm:icf_}, with probability at least \(1 - \delta/2\) we have that \(\pr{x \sim \calD}{s(x)=0} \leq (\theta+\epsilon')/(R-1)\). Also, from Hoeffding's inequality, since \( \abs{ X_\mathrm{test} } \gtrsim \log(1/\delta)/\epsilon^2 \), we have that with probability at least \(1 - \delta/2\), it holds \(\pr{x \sim \calU(X_\mathrm{test})}{s(x)=0} \leq \pr{x \sim \calD'}{s(x)=0} + \epsilon/8\). Hence, with probability at least \(1 - \delta\), it holds
    \[
        \pr{x \sim \calU(X_\mathrm{test})}{s(x)=0} \leq \pr{x \sim \calD}{s(x)=0} + \theta + \epsilon/8 \leq \theta + \dfrac{\theta + \epsilon'}{R-1} + \epsilon/8 < \dfrac{R\theta}{R-1} + \epsilon/4 \, .
    \]
    Thus, if \(d_\mathrm{TV}(\calD, \calD') \leq \theta\), then the algorithm accepts with probability at least \(1 - \delta\).
\end{proof}

\section{Additional Tools}

\begin{theorem}[Implicit in \cite{doi:10.1137/060649057}] \label{thm:l1_regression}
    Let any \(d \in \natural^\ast\), \(\ell \in \natural\), let \(\calX \subseteq \real^d\), let \(\calD\) be a distribution on \(\calX \times \{0, 1\}\) and let \(\calD_{\calX}\) be the marginal of \(\calD\) on \(\calX\). Let any \(\epsilon, \delta \in (0,1)\). Let \(\calC \subseteq \{0, 1\}^\calX\) such that for all \(f \in \calC\) there exists a polynomial \(p \in \calP_{d, \ell}\) such that \(\ev{x \sim \calD_{\calX}}{\abs{p(x) - f(x)}} \leq \epsilon/2\).
    Then, the degree-\(\ell\) \(\calL_1\) polynomial regression algorithm of \cite{doi:10.1137/060649057} agnostically learns \(\calC\) up to error \(\inf_{f \in \calC}\pr{(x,y) \sim \calD}{f(x) \neq y} + \epsilon\) and failure probability \(\delta\) with \(\mathrm{poly} \left((d+1)^\ell, 1/\epsilon, \log(1/\delta) \right)\) sample complexity and runtime.
\end{theorem}

\begin{lemma}[Lemma A.1 from \cite{chandrasekaran2025learning}] \label{lemma:loewner_conc}
    Let any \(d \in \natural^\ast, \ell \in \natural\) and \(A \geq 1\). Let \(\calD\) be an \(A\)-hypercontractive distribution on \(\real^d\). Let \(S\) be a multiset of \(m\) independent examples from \(\calD\). Then, for any \(\delta \in (0, 1)\), if \(m \gtrsim (d+1)^{2\ell}   \left(4A\log_2(4/\delta) \right)^{4\ell+1}\), then with probability at least \(1-\delta\), we have that for all \(p \in \calP_{d, \ell}\), it holds
    \[
    \dfrac{1}{2} \ev{x \sim \calD}{p(x)^2} \leq  \ev{x \sim \calU(S)}{p(x)^2} \leq
    2 \ev{x \sim \calD}{p(x)^2} \, .
    \]
\end{lemma}

\begin{lemma}[Lemma A.2 from \cite{klivans2026poweriterativefilteringsupervised}] \label{lemma:con1}
    Let any \(d \in \natural^\ast, \ell \in \natural\), \(A \geq 1\) and \(\beta > 0\). Let \(\calD\) be an \(A\)-hypercontractive distribution on \(\real^d\). Let \(S\) be a multiset of \(m\) independent examples from \(\calD\) and let \(\calP = \{p \in \calP_{d, \ell} : \ev{x \sim \calU(S)}{p(x)^2 } \leq 2\beta\}\). Then, for any \(\epsilon, \delta \in (0, 1)\), if \(m \gtrsim (d+1)^{2\ell}   \left(4A\log_2(4/\delta) \right)^{4\ell+1}\), then with probability at least \(1-\delta\), it holds
    \[
    \PR{x \sim \calD}{\max_{p \in \calP} \abs{p(x)} > 2\sqrt{\dfrac{\beta(d+1)^\ell}{\epsilon}} } \leq \epsilon \, .
    \]
\end{lemma}

\begin{lemma} \label{lemma:con2}
    Let any \(d \in \natural^\ast, \ell \in \natural\), \(A \geq 1\), \(\epsilon, \beta > 0\) and \(R \geq 1\). Let \(\calD\) be an \(A\)-hypercontractive distribution on \(\real^d\). Let any \(p \in \calP_{d, \ell}\) and \(f \colon \real^d \to \{0, 1\}\) such that \(\ev{x \sim \calD}{f(x) \abs{p(x)}} \leq \epsilon/(R+\epsilon)\) and \(\ev{x \sim \calD}{p(x)^2} \leq \beta\). Let \(S\) be a multiset of \(m\) independent examples from \(\calD\). Then, for any \(\delta \in (0, 1)\), if \(m \gtrsim R^4\beta(d+1)^{2\ell} \left(4A\log_2(8/\delta) \right)^{4\ell+1}/\epsilon^4\), then with probability at least \(1 - \delta\), we have that \(\ev{x \sim \calU(S)}{f(x) \abs{p(x)}} \leq 2\epsilon/(2R+\epsilon)\) and \(\ev{x \sim \calU(S)}{p(x)^2} \leq 2\beta\).
\end{lemma}

\begin{proof}
    If \(m\) is sufficiently large, by \Cref{lemma:loewner_conc}, we have that with probability at least \(1 - \delta/2\), it holds \(\ev{x \sim \calU(S)}{p(x)^2} \leq 2\beta\). It remains to show that \(\ev{x \sim \calU(S)}{f(x) \abs{p(x)}} \leq 2\epsilon/(2R+\epsilon)\) with probability at least \(1 - \delta/2\). For any \(z \in \real^d\), let \(\psi(z) = f(z)\abs{p(z)} - \ev{x \sim \calD}{f(x)\abs{p(x)}}\). Then, for any \(i \geq 1\) and for any \(\varepsilon > 0\), it holds\footnote{The notation \(S \sim \calD^m\) means that \(S\) is a multiset of \(m\) independent samples from \(\calD\).}
    {
    \allowdisplaybreaks
    \begin{align*}
        \PR{S \sim \calD^{m}}{\ABS{\ev{x \sim \calU(S)}{\psi(x)}} > \varepsilon}
    &\leq \dfrac{1}{\varepsilon^{2i}} \EV{S \sim \calD^{m}}{\ABS{ \ev{x \sim \calU(S)}{\psi(x)} }^{2i}} \\
    &\leq \left(\dfrac{4i}{\varepsilon^2 m} \right)^{i} \ev{x \sim \calD}{\abs{ \psi(x) }^{2i}} \\
    &\leq 2\left(\dfrac{16i}{\varepsilon^2 m} \right)^{i} \ev{x \sim \calD}{\abs{p(x)}^{2i}} \\
    &\leq 2 (2Ai)^{2\ell i}\left(\dfrac{16i}{\varepsilon^2 m} \right)^{i} \ev{x \sim \calD}{\abs{p(x)}}^{2i} \\
    &\leq 2 \left( \dfrac{16 (2A)^{2\ell}i^{2\ell+1}\beta}{\varepsilon^2m} \right)^i
    \, .
    \end{align*}
    }
    In the first line, we used Markov's inequality. In the second line, we used the Marcinkiewicz–Zygmund inequality and Hölder's inequality (see \cite{FERGER201496}). In the third line, we used the fact that \((a+b)^{2k} \leq 4^k \max\{a^{2k}, b^{2k}\}\) for any \(a, b, k \geq 0\). Let \(\varepsilon = 2\epsilon/(2R+\epsilon)-\epsilon/(R+\epsilon)>\epsilon^2/6R^2\). If \(m \gtrsim R^4\beta(d+1)^{2\ell} \left(4A\log_2(8/\delta) \right)^{4\ell+1}/\epsilon^4\), then, for \(i = \lceil\log_2(4/\delta)\rceil\), we get that the term in the last line is upper bounded by \(\delta/2\), which concludes the proof.
\end{proof}

\begin{lemma} \label{lemma:vc1}
    Let any \(d, \ell \in \natural^\ast\). Let \(\calF \subset \{\real^d \to \{0, 1\}\}\) be a finite class and let \(\calH = \{x \mapsto \indicator{f(x) \abs{p(x)} > \tau} : p \in \calP_{d, \ell}, \tau \geq 0, f \in \calF\}\). Then, \(\mathrm{VCdim}(\calH) \lesssim (d+1)^\ell \ell \log(d) + \log(\abs{\calF})\).
\end{lemma}

\begin{proof}
    For any \(f \in \calF\), let \(\calH_f = \{x \mapsto \indicator{f(x) \abs{p(x)} > \tau} : p \in \calP_{d, \ell}, \tau \geq 0\}
    \). We will first show that \(\mathrm{VCdim}(\calH_f) \lesssim (d+1)^\ell\) for any \(f \in \calF\). Let \(\calG = \{x \mapsto \indicator{\abs{p(x)} > \tau} : p \in \calP_{d, \ell}, \tau \geq 0\}\). Notice that for any \(x \in \real^d\), \(p \in \calP_{d, \ell}\) and \(\tau \geq 0\), it holds \(\indicator{ \abs{p(x)} > \tau} = \indicator{ p(x) > \tau} \lor \indicator{ p(x) < -\tau}\), namely, the function  \(x \mapsto \indicator{ \abs{p(x)} > \tau}\) is the logical OR of two PTFs of degree at most \(\ell\). Thus, by \Cref{fact:vc_and}, it holds \(\mathrm{VCdim}(\calG) \lesssim (d+1)^\ell\). Finally, notice that \(\calH_f = \{ fg : g \in \calG \}\), which implies that \(\mathrm{VCdim}(\calH_f) \leq \mathrm{VCdim}(\calG) \lesssim (d+1)^\ell\).
    Since \(\calH = \bigcup_{f \in \calF} \calH_f\), by \Cref{fact:vc_union}, we deduce that \(\mathrm{VCdim}(\calH) \lesssim (d+1)^\ell \ell \log(d) + \log(\abs{\calF})\).
\end{proof}

\begin{lemma} \label{lemma:uc1}
    Let any \(M > 0\). Let any set \(\calX\), any distribution \(\calD\) on \(\calX\), any \(f \colon \calX \to \real\) and any \(\calG \subseteq \{0, 1\}^\calX\) such that \(\sup_{g \in \calG}\sup_{x \in \calX} g(x) \abs{f(x)} \leq M\) and \(\mathrm{VCdim}(\calG) < \infty\). Let \(S\) be a multiset of \(n\) i.i.d. samples drawn from \(\calD\). Then, there exists some universal constant \(C > 0\) such that for any \(\delta \in (0, 1)\), with probability at least \(1 - \delta\) over \(S\), it holds
    \[
    \sup_{g \in \calG} \ABS{ \ev{x \sim \calU(S)}{f(x)g(x)} - \ev{x \sim \calD}{f(x)g(x)} } \leq CM\left( \dfrac{\mathrm{VCdim}(\calG) \log(n) + \log(1/\delta)}{n}  \right)^{1/4} \, .
    \]
\end{lemma}

\begin{proof}
    Let any \(\Delta > 0\) such that \(k:=M/\Delta\) is an integer. For any \(i \in \integer\) and \(x \in \calX\), we define \(I_i(x) := \indicator{i \Delta \leq f(x) < (i+1)\Delta}\). Then, for any \(g \in \calG\), we have
    {
    \allowdisplaybreaks
    \begin{align*}
        \ABS{ \ev{x \sim \calU(S)}{f(x)g(x)} - \ev{x \sim \calD}{f(x)g(x)} } &= 
        \ABS{ \sum_{i=-k}^{k} \ev{x \sim \calU(S)}{f(x)g(x) I_i(x)} - \ev{x \sim \calD}{f(x)g(x) I_i(x)} } \\&\leq 
        \sum_{i=-k}^{k} \ABS{ \ev{x \sim \calU(S)}{f(x)g(x) I_i(x)} - \ev{x \sim \calD}{f(x)g(x) I_i(x)} } \\&\leq 
        \sum_{i=-k}^{k}  \ABS{ i\Delta \left(  \ev{x \sim \calU(S)}{g(x) I_i(x)} - \ev{x \sim \calD}{g(x) I_i(x)} \right) } \\&+
        \Delta\ev{x \sim \calU(S)}{g(x) I_i(x)} + \Delta\ev{x \sim \calD}{g(x) I_i(x)}
        \\&\leq 
        2\Delta +  \sum_{i=-k}^{k}  \ABS{ i\Delta \left(  \ev{x \sim \calU(S)}{g(x) I_i(x)} - \ev{x \sim \calD}{g(x) I_i(x)} \right) } \\&\leq 
        2\Delta + M\sum_{i=-k}^{k} \ABS{ \ev{x \sim \calU(S)}{g(x) I_i(x)} - \ev{x \sim \calD}{g(x) I_i(x)} } .
    \end{align*}
    }
    Since the VC dimension of the class \(\{x \mapsto \indicator{a \leq f(x) < b} : a, b \in \real \}\) is at most \(2\), by \Cref{fact:vc_and}, we have that for any \(g \in \calG\), integer \(i\) and \(\Delta > 0\), the function \(g(x) I_i(x)\) belongs to a class of VC dimension at most \(O\left(\mathrm{VCdim}(\calG)\right)\).
    Thus, by \Cref{fact:uc0}, we have that with probability at least \(1 - \delta\), it holds
    \[
    \sup_{g \in \calG} \ABS{ \ev{x \sim \calU(S)}{f(x)g(x)} - \ev{x \sim \calD}{f(x)g(x)} } \leq  O \left( \Delta +  \dfrac{M^2}{\Delta}\sqrt{\dfrac{\mathrm{VCdim}(\calG) \log(n) + \log(1/\delta)}{n}} \right) \, .
    \]
    Taking \(\Delta\) to minimize the right hand side, concludes the proof.
\end{proof}

\end{document}